\documentclass[11pt,letterpaper,twoside,reqno]{amsart}

\usepackage{fixltx2e}

\usepackage{amsthm}
\usepackage{amsmath}
\usepackage{amssymb}
\usepackage{mathrsfs}
\usepackage[dvips]{graphicx}
\usepackage{mathtools}

\mathtoolsset{showonlyrefs}

\usepackage{wasysym}

\usepackage{bbm}

\usepackage{bm}

\usepackage{color}

\definecolor{dark-gray}{gray}{0.3}
\definecolor{dkblue}{rgb}{0,0,0.3}
\definecolor{dkred}{rgb}{0.5,0,0}

\usepackage[colorlinks=true]{hyperref}
\hypersetup{urlcolor=dkred}
\hypersetup{citecolor=dkblue}
\hypersetup{linkcolor=dkblue}

\usepackage{url}

\usepackage{supertabular}

\newtheorem{thm}{Theorem}
\newtheorem{cor}[thm]{Corollary}
\newtheorem{lemma}[thm]{Lemma}
\newtheorem{prop}[thm]{Proposition}

\newtheorem{defn}[thm]{Definition}

\theoremstyle{remark}

\numberwithin{equation}{section}

\numberwithin{thm}{section}

\newcommand{\term}{\emph}

\renewcommand{\phi}{\varphi}

\newcommand{\econst}{\mathrm{e}}

\newcommand{\Id}{\mathbf{I}}

\newcommand{\zeromtx}{\bm{0}}

\newcommand{\coll}[1]{\mathscr{#1}}

\newcommand{\R}{\mathbb{R}}
\newcommand{\C}{\mathbb{C}}

\newcommand{\M}{\mathbb{M}}
\newcommand{\Sym}{\mathbb{H}}

\newcommand{\B}{\mathbb{B}}

\newcommand{\diff}[1]{\mathrm{d}{#1}}
\newcommand{\idiff}[1]{\, \diff{#1}}

\newcommand{\Prob}[1]{\mathbb{P}\left\{ {#1} \right\}}
\newcommand{\Expect}{\operatorname{\mathbb{E}}}

\newcommand{\vct}[1]{\bm{#1}}
\newcommand{\mtx}[1]{\bm{#1}}

\newcommand{\adj}{*}

\newcommand{\diag}{\operatorname{diag}}
\newcommand{\trace}{\operatorname{tr}}

\newcommand{\psdle}{\preccurlyeq}

\newcommand{\ip}[2]{\left\langle {#1},\, {#2} \right\rangle}

\newcommand{\norm}[1]{\left\Vert {#1} \right\Vert}

\newcommand{\ntr}{\operatorname{\bar{tr}}}
\newcommand{\condl}{\, \vert \, }

\newcommand{\op}[1]{\mathsf{#1}}

\begin{document}

\title[Subadditivity of Matrix $\phi$-Entropy]
{Subadditivity of Matrix $\phi$-Entropy \\ and Concentration of Random Matrices}
\author{Richard~Y.~Chen and Joel A. Tropp }
\date{13 August 2013}

\maketitle
\begin{abstract}
Matrix concentration inequalities provide a direct way to bound the typical spectral norm of a random matrix.  The methods for establishing these results often parallel classical arguments, such as the Laplace transform method.  This work develops a matrix extension of the entropy method, and it applies these ideas to obtain some matrix concentration inequalities.
\end{abstract}

\section{Introduction}

Matrix concentration inequalities offer a direct way to bound the spectral-norm deviation of a random  matrix from its mean~\cite{Ahlswede2002,Oliveira2009,Tropp2011,Tropp2012,Minsker2012,Mackey2012,Paulin2013}.  Owing to their simplicity and ease of use, these results have already found applications in a wide range of areas, including random graph theory \cite{Oliveira2009, Chaudhuri2012}, randomized linear algebra \cite{Drineas2011,Tro11:Improved-Analysis,Boutsidis2012}, least-squares approximation \cite{Cohen2011}, compressed sensing \cite{Ahmed2012,Tang2012}, and matrix completion \cite{Recht2011, Mackey2011, Koltchinskii2011}.

The most effective methods for developing matrix concentration inequalities parallel familiar scalar arguments.  For example, it is possible to mimic the Laplace transform method of Bernstein to obtain powerful results for sums of independent random matrices~\cite{Ahlswede2002,Oliveira2010,Tropp2012}.
Several other papers have adapted martingale methods to the matrix setting~\cite{Oliveira2009,Tropp2011,Minsker2012}.  A third line of work~\cite{Mackey2012,Paulin2013} contains a matrix extension of Chatterjee's techniques~\cite{Chatterjee2007,Chatterjee2008} for proving concentration inequalities via Stein's method of exchangeable pairs.  See the survey~\cite{Tro12:User-Friendly} for a more complete bibliography.

In spite of these successes, the study of matrix concentration inequalities is by no means complete.  Indeed, one frequently encounters random matrices that do not submit to existing techniques.  The aim of this paper is to explore the prospects for adapting $\phi$-Sobolev inequalities~\cite{Latala2000,Chafai2004,Boucheron2005} to the matrix setting.  By doing so, we hope to obtain concentration inequalities that hold for general matrix-valued functions of independent random variables.

It is indeed possible to obtain matrix analogs of the scalar $\phi$-Sobolev inequalities for product spaces that appear in~\cite{Boucheron2005}.  This theory leads to some interesting concentration inequalities for random matrices.  On the other hand, this method is not as satisfying as some other approaches to matrix concentration because the resulting bounds seem to require artificial assumptions.  Nevertheless, we believe it is worthwhile to document the techniques and to indicate where matrix $\phi$-Sobolev inequalities differ from their scalar counterparts.

\subsection{Notation and Background} \label{subsection:notation}

Before we can discuss our main results, we must instate some notation.  The set $\R_+$ contains the nonnegative real numbers, and $\R_{++}$ consists of all positive real numbers.  We write $\mathbb{M}^d$ for the complex Banach space of $d\times d$ complex matrices, equipped with the usual $\ell_2$ operator norm $\norm{\cdot}$.  The \term{normalized trace} is the function
$$
\ntr \mtx{B} := \frac{1}{d} \sum\nolimits_{j=1}^d b_{jj}
\quad\text{for $\mtx{B} \in \M^d$.}
$$
The theory can be developed using the standard trace, but additional complications arise.

The set $\Sym^d$ refers to the real-linear subspace of $d \times d$ Hermitian matrices in $\M^d$.  For a matrix $\mtx{A} \in \Sym^d$, we write $\lambda_{\min}(\mtx{A})$ and $\lambda_{\max}(\mtx{A})$ for the algebraic minimum and maximum eigenvalues.  For each interval $I \subset \R$, we define the set of Hermitian matrices whose eigenvalues fall in that interval:
\begin{equation*}
\Sym^d(I): = \{\mtx{A} \in \Sym^d: [\lambda_{\min}(\mtx{A}), \lambda_{\max}(\mtx{A})] \subset I \}.
\end{equation*}
We also introduce the set $\Sym^d_+$ of $d \times d$ positive-semidefinite matrices and the set $\Sym^d_{++}$ of $d \times d$ positive-definite matrices.  Curly inequalities refer to the positive-semidefinite order.  For example, $\mtx{A} \psdle \mtx{B}$ means that $\mtx{B} - \mtx{A}$ is positive semidefinite.

Next, let us explain how to extend scalar functions to matrices.  Recall that each Hermitian matrix $\mtx{A} \in \Sym^d$ has a \term{spectral resolution}
\begin{equation} \label{eqn:spectral}
\mtx{A} = \sum\nolimits_{i=1}^d \lambda_i \mtx{P}_i,
\end{equation}
where $\lambda_1, \dots, \lambda_d$ are the eigenvalues of $\mtx{A}$ and the matrices $\mtx{P}_1, \dots, \mtx{P}_d$ are orthogonal projectors that satisfy the orthogonality relations
\begin{equation*}
\mtx{P}_i \mtx{P}_j = \delta_{ij} \mtx{P}_j \quad \text{and}\quad \sum\nolimits_{i=1}^d \mtx{P}_i = \Id,
\end{equation*}
where $\delta_{ij}$ is the Kronecker delta and $\Id$ is the identity matrix.
One obtains a standard matrix function by applying a scalar function to the spectrum of a Hermitian matrix.

\begin{defn}[Standard Matrix Function] \label{defn:matrix-function}
Let $f: I \mapsto \R$ be a function on an interval $I$ of the real line.  Suppose that $\mtx{A} \in \Sym^d(I)$ has the spectral decomposition \eqref{eqn:spectral}. Then
\begin{equation*}
f(\mtx{A}) : = \sum\nolimits_{i=1}^d f(\lambda_i) \mtx{P}_i.
\end{equation*}
\end{defn}

\noindent
We use lowercase Roman and Greek letters to refer to standard matrix functions.  When we apply a familiar real-valued function to an Hermitian matrix, we are referring to the associated standard matrix function.  Bold capital letters such as $\mtx{Y}, \mtx{Z}$ denote general matrix functions that are not necessarily standard.

\subsection{Subadditivity of Matrix Entropies}

In this section, we provide an overview of the theory of matrix $\phi$-entropies.  Our entire approach has a strong parallel with the work of Boucheron et al.~\cite{Boucheron2005}.  In the matrix setting, however, the technical difficulties
are more formidable.

\subsubsection{The Class of Matrix Entropies}

First, we carve out a class of standard matrix functions that we can use to construct matrix entropies with the same subadditivity properties as their scalar counterparts.

\begin{defn} [$\Phi_{d}$ Function Class] \label{defn:function-intro}
Let $d$ be a natural number.  The class $\Phi_d$ contains each function $\phi : \R_+ \to \R$ that is either affine or else satisfies the following three conditions.
\begin{enumerate}
\item $\phi$ is continuous and convex.
\item $\phi$ has two continuous derivatives on $\R_{++}$.
\item Define $\psi(t) = \phi'(t)$ for $t \in \R_{++}$.
The derivative $\op{D}\psi$ of the standard matrix function $\psi : \Sym^d_{++} \to \Sym^d$ is an invertible linear operator on $\Sym_{++}^d$, and the map $\mtx{A} \mapsto [\op{D}\psi(\mtx{A})]^{-1}$ is concave with respect to the semidefinite order on operators.
\end{enumerate}
\end{defn}

\noindent
The technical definitions that support requirement (3) appear in Section~\ref{section:preliminaries}.  For now, we just remark that the scalar equivalent of (3) is the statement that $t \mapsto [\phi''(t)]^{-1}$ is concave on $\R_{++}$.

The class $\Phi_1$ coincides with the $\Phi$ function class considered in~\cite{Boucheron2005}.  It can be shown that $\Phi_{d+1} \subseteq \Phi_{d}$ for each natural number $d$, so it is appropriate to introduce the \term{class of matrix entropies}:
$$
\Phi_{\infty} := \bigcap\nolimits_{d=1}^\infty \Phi_d
$$
This class consists of scalar functions that satisfy the conditions of Definition~\ref{defn:function-intro} for an arbitrary choice of dimension $d$.
Note that $\Phi_{\infty}$ is a convex cone: it contains all positive multiples and all finite sums of its elements.

In contrast to the scalar setting, it is not easy to determine what functions are contained in $\Phi_{\infty}$.  The main technical achievement of this paper is to demonstrate that the standard entropy and certain power functions belong to the matrix entropy class.

\begin{thm}[Elements of the Matrix Entropy Class] \label{thm:membership}
The following functions are members of the $\Phi_{\infty}$ class.
\begin{enumerate}
\item	The standard entropy $t \mapsto t\log t$.

\item	The power function $t \mapsto t^p$ for each $p \in [1,2]$.
\end{enumerate}
\end{thm}

\noindent
The statement about the classical entropy can be obtained from standard results in matrix theory, but the statement for power functions demands some effort.
The proof of Theorem~\ref{thm:membership} appears in Section \ref{section:member}.

\subsubsection{Matrix \texorpdfstring{$\phi$}{phi}-Entropy}

For each function in the matrix entropy class, we can introduce a generalized entropy functional that measures the amount of fluctuation in a random matrix.

\begin{defn}[Matrix $\phi$-Entropy] \label{defn:entropy-intro}
Let $\phi \in \Phi_{\infty}$.  Consider a random matrix $\mtx{Z}$ taking values in $\Sym_+^d$, and assume that $\Expect \norm{ \mtx{Z} } < \infty$ and $\Expect \norm{\phi(\mtx{Z})} < \infty$.  The matrix $\phi$-entropy functional $H_{\phi}$ is
\begin{equation} \label{eqn:entropy-intro}
H_\phi(\mtx{Z}) := \Expect  \ntr \phi(\mtx{Z}) - \ntr \phi( \Expect \mtx{Z}).
\end{equation}
Similarly, the conditional matrix $\phi$-entropy functional is
\begin{equation*}
H_{\phi}(\mtx{Z} \condl \coll{F}) := \Expect\big[ \ntr \phi(\mtx{Z}) \condl \coll{F} \big]
	- \ntr \phi\big(\Expect [\mtx{Z} \condl \coll{F} ] \big),
\end{equation*}
where $\coll{F}$ is a subalgebra of the master sigma algebra.  
\end{defn}

For each convex function $\phi$, the trace function $\ntr \phi : \Sym^d_+ \to \R$ is also convex~\cite[Sec.~2.2]{Carlen2010}.  Therefore, Jensen's inequality implies that the matrix $\phi$-entropy is nonnegative:
\begin{equation*} %
H_{\phi}(\mtx{Z}) \geq 0.
\end{equation*}
For concreteness, here are some basic examples of matrix $\phi$-entropy functionals.
\begin{align*}
&H_{\phi}(\mtx{Z}) = \ntr \big[\Expect (\mtx{Z} \log \mtx{Z} )- (\Expect \mtx{Z}) \log (\Expect \mtx{Z}) \big]
	&&\text{when $\phi(t)= t \log t$.} \\
&H_{\phi}(\mtx{Z}) = \ntr \big[\Expect (\mtx{Z}^p)- (\Expect \mtx{Z})^p \big]
	&&\text{when $\phi(t) = t^p$ for $p \in [1, 2].$} \\
&H_{\phi}(\mtx{Z}) = 0
	&&\text{when $\phi$ is affine.}
\end{align*}

\subsubsection{Subadditivity of Matrix \texorpdfstring{$\phi$}{phi}-Entropy} \label{subsection:subadditivity}

The key fact about matrix $\phi$-entropies is that they satisfy a subadditivity property.
Let $\vct{x} := (X_1, \dots, X_n)$ denote a vector of independent random variables taking values in a Polish space, and write $\vct{x}_{-i}$ for the random vector obtained by deleting the $i$th entry of $\vct{x}$.
\begin{equation*}
\vct{x}_{-i} := (X_1, \dots, X_{i-1}, X_{i+1}, \dots, X_n).
\end{equation*}
Consider a positive-semidefinite random matrix $\mtx{Z}$ that can be expressed as a measurable function of the random vector $\vct{x}$.
$$
\mtx{Z} := \mtx{Z}(X_1, \dots, X_n) \in \Sym_+^d.
$$
We instate the integrability conditions $\Expect \norm{\mtx{Z}} < \infty$ and $\Expect \norm{\phi(\mtx{Z})} < \infty$.

\begin{thm}[Subadditivity of Matrix $\phi$-Entropy] \label{thm:subadditivity-intro}
Fix a function $\phi \in \Phi_{\infty}$.  Under the prevailing assumptions,
\begin{equation} \label{eqn:subadditivity-intro}
H_\phi( \mtx{Z}) \leq  \sum\nolimits_{i=1}^n \Expect \left[  H_{\phi}(\mtx{Z} \condl \vct{x}_{-i}) \right].
\end{equation}
\end{thm}

Typically, we apply Theorem~\ref{thm:subadditivity-intro} by way of a corollary.
Let $X'_1, \dots, X'_n$ denote independent copies of $X_1, \dots, X_n$, and form the random matrix
\begin{equation} \label{eqn:Z-i}
\mtx{Z}'_i := \mtx{Z}(X_1, \dots, X_{i-1}, X'_i, X_{i+1}, \dots, X_n) \in \Sym_+^d.
\end{equation}
Then $\mtx{Z}'_i$ and $\mtx{Z}$ are independent and identically distributed, conditional on the sigma algebra generated by $\vct{x}_{-i}$.  In particular, these two random matrices are exchangeable counterparts.

\begin{cor}[Entropy Bounds via Exchangeability] \label{cor:subadditivity-intro}
Fix a function $\phi \in \Phi_\infty$, and write $\psi = \phi'$.
With the prevailing notation,
\begin{equation} \label{eqn:cor-subadditivity-intro}
H_\phi( \mtx{Z}) \leq \frac{1}{2} \sum\nolimits_{i=1}^n \Expect  \ntr\Big[ (\mtx{Z} - \mtx{Z}'_i )(\psi(\mtx{Z}) - \psi(\mtx{Z}'_i)) \Big].
\end{equation}
\end{cor}

Theorem \ref{thm:subadditivity-intro} and Corollary~\ref{cor:subadditivity-intro} are matrix counterparts of the foundational results from Boucheron et al.~\cite[Sec.~3]{Boucheron2005}, which establish that scalar $\phi$-entropies satisfy a similar subadditivity property.  We devote Section \ref{section:subadditivity} to the proof of these results.

\subsection{Some Matrix Concentration Inequalities}

Using Corollary~\ref{cor:subadditivity-intro}, we can derive concentration inequalities for random matrices.  In contrast to some previous approach to matrix concentration, we need to place some significant restrictions on the type of random matrices we consider.

\begin{defn}[Invariance under Signed Permutation] \label{def:signed-permutation}
A random matrix $\mtx{Y} \in \Sym^d$ is invariant under signed permutation if we have the equality of distribution
$$
\mtx{Y} \sim \mtx{\Pi}^* \mtx{Y} \mtx{\Pi}
\quad\text{for each signed permutation $\mtx{\Pi}$.}
$$
A \term{signed permutation} $\mtx{\Pi} \in \M^d$ is a matrix with the properties that (i) each row and each column contains exactly one nonzero entry and (ii) the nonzero entries only take values $+1$ and $-1$.
\end{defn}

\subsubsection{A Bounded Difference Inequality}

We begin with an exponential tail bound for a random matrix whose distribution is invariant under signed permutations.

\begin{thm}[Bounded Differences] \label{thm:tail-bound}
Let $\vct{x} := (X_1, \dots, X_n)$ be a vector of independent random variables, and let $\vct{x}' := (X'_1, \dots, X'_n)$ be an independent copy of $\vct{x}$.
Consider random matrices
\begin{align*}
\mtx{Y} &:= \mtx{Y}(X_1, \dots, X_i, \dots, X_n) \in \Sym^d
\quad\text{and} \\
\mtx{Y}_i' &:= \mtx{Y}(X_1, \dots, X_i', \dots, X_n) \in \Sym^d
\quad\text{for $i = 1, \dots, n$.}
\end{align*}
Assume that $\norm{\mtx{Y}}$ is bounded almost surely.  Introduce the variance measure
\begin{equation} \label{eqn:scalar-variance}
V_{\mtx{Y}} := \sup \norm{ \Expect \left[ \sum\nolimits_{i=1}^n (\mtx{Y} - \mtx{Y}_i')^2 \, \Big\vert \, \vct{x} \right] },
\end{equation}
where the supremum occurs over all possible values of $\vct{x}$.  For each $t \geq 0$,
\begin{align*}
\Prob{\lambda_{\max}(\mtx{Y}- \Expect \mtx{Y}) \geq t} &\leq
	d \cdot \econst^{-t^2/ (2 V_{\mtx{Y}})},
	\quad\text{and} \\
\Prob{\lambda_{\min}(\mtx{Y}- \Expect \mtx{Y}) \leq -t} &\leq
	d \cdot \econst^{-t^2/ (2 V_{\mtx{Y}})}.
\end{align*} 
\end{thm}

\noindent
Theorem~\ref{thm:tail-bound} follows from Corollary~\ref{cor:subadditivity-intro} with the choice $\phi(t) = t \log t$.  See Section~\ref{section:concentration} for the proof.  This result can be viewed as a type of matrix bounded difference inequality.  Closely related inequalities already appear in the literature; see~\cite[Cor.~7.5]{Tropp2012},~\cite[Cor.~11.1]{Mackey2012}, and~\cite[Cor.~4.1]{Paulin2013}.  Theorem~\ref{thm:tail-bound} is weaker than some of the earlier bounds because of the restriction to random matrices that are invariant under signed permutation.

\subsubsection{Matrix Moment Bounds}

We can also establish moment inequalities for a random matrix whose distribution is invariant under signed permutation.

\begin{thm}[Matrix Moment Bound]  \label{thm:moment-bound-intro}
Fix a number $q \in \{2, 3, 4, \dots \}$.  Let $\vct{x} := (X_1, \dots, X_n)$ be a vector of independent random variables, and let $\vct{x}' := (X'_1, \dots, X'_n)$ be an independent copy of $\vct{x}$.  Consider positive-semidefinite random matrices
\begin{align*}
\mtx{Y} &:= \mtx{Y}(X_1, \dots, X_i, \dots, X_n) \in \Sym_+^d
\quad\text{and} \\
\mtx{Y}_i' &:= \mtx{Y}(X_1, \dots, X_i', \dots, X_n) \in \Sym_+^d
\quad\text{for $i = 1, \dots, n$.}
\end{align*}
Assume that $\Expect( \norm{\mtx{Y}}^q ) < \infty$.
Suppose that there is a constant $c \geq 0$ with the property
\begin{equation} \label{eqn:matrix-variance}
\mtx{V}_{\mtx{Y}} :=
\Expect \left[ \sum\nolimits_{i=1}^n (\mtx{Y} - \mtx{Y}_i')^2 \, \Big\vert \, \vct{x} \right]
	\psdle c \, \mtx{Y}.
\end{equation}
Then the random matrix $\mtx{Y}$ satisfies the moment inequality
\begin{equation*}
[\Expect \ntr( \mtx{Y}^q ) ]^{1/q} \leq \Expect \ntr \mtx{Y} + \frac{q-1}{2} \cdot c.
\end{equation*}
\end{thm}

Theorem~\ref{thm:moment-bound-intro} follows from Corollary~\ref{cor:subadditivity-intro} with the choice $\phi(t) = t^{q/(q-1)}$.  See Section~\ref{section:moment} for the proof.
This result can be regarded as a matrix extension of a moment inequality for real random variables~\cite[Cor.~1]{Boucheron2005}.  The papers~\cite{Mackey2012,Paulin2013} contain similar moment inequalities for random matrices.  See also~\cite{JX03:Noncommutative-Burkholder,JX08:Noncommutative-Burkholder-II,JZ11:Noncommutative-Bennett}.

\subsection{Generalized Subadditivity of Matrix $\phi$-Entropy}

Theorem~\ref{thm:subadditivity-intro} is the shadow of a more sophisticated subadditivity property.  We outline the simplest form of this more general result.  See the lecture notes of Carlen~\cite{Carlen2010} for more background on the topics in this section.

We work in the $*$-algebra $\M^d$ of $d\times d$ complex matrices, equipped with the conjugate transpose operation $*$ and the normalized trace inner product $\ip{\mtx{A}}{\mtx{B}} := \ntr(\mtx{A}^\adj \mtx{B})$.  We say that a subspace $\mathfrak{A} \subset \M^d$ is a \term{$*$-subalgebra} when $\mathfrak{A}$ contains the identity matrix, $\mathfrak{A}$ is closed under matrix multiplication, and $\mathfrak{A}$ is closed under conjugate transposition.  In other terms, $\Id \in \mathfrak{A}$ and $\mtx{AB} \in \mathfrak{A}$ and $\mtx{A}^\adj \in \mathfrak{A}$ whenever $\mtx{A}, \mtx{B} \in \mathfrak{A}$.

In this setting, there is an elegant notion of conditional expectation.  The orthogonal projector $\Expect_{\mathfrak{A}} : \M^d \to \mathfrak{A}$ onto the $*$-subalgebra $\mathfrak{A}$ is called the \term{conditional expectation} with respect to the $*$-subalgebra.  For $*$-subalgebras $\mathfrak{A}$ and $\mathfrak{B}$, we say that the conditional expectations $\Expect_{\mathfrak{A}}$ and $\Expect_{\mathfrak{B}}$ \term{commute} when
$$
(\Expect_{\mathfrak{A}} \Expect_{\mathfrak{B}}) (\mtx{M})
	= (\Expect_{\mathfrak{B}} \Expect_{\mathfrak{A}}) (\mtx{M})
	\quad\text{for every $\mtx{M} \in \M^d$.}
$$
This construction generalizes the concept of independence in a probability space.

We can define the matrix $\phi$-entropy conditional on a $*$-subalgebra $\mathfrak{A}$:
$$
H_{\phi}(\mtx{A} \condl \mathfrak{A}) :=
	\ntr[ \phi(\mtx{A}) - \phi( \Expect_{\mathfrak{A}}\mtx{A} ) ]
	\quad\text{for $\mtx{A} \in \Sym^d_+$.}
$$
Let $\mathfrak{A}_1, \dots, \mathfrak{A}_n$ be $*$-subalgebras whose conditional expectations commute.  Then we can extend the definition of the matrix $\phi$-entropy to read
$$
H_{\phi}(\mtx{A} \condl \mathfrak{A}_1, \dots, \mathfrak{A}_n)
	:= \ntr[ \phi(\mtx{A})
	- \phi( \Expect_{\mathfrak{A}_1} \cdots \Expect_{\mathfrak{A}_n} \mtx{A} ) ]
	\quad\text{for $\mtx{A} \in \Sym^d_+$.}
$$
Because of commutativity, the order of the conditional expectations has no effect on the calculation.  It turns out that matrix $\phi$-entropy admits the following subadditivity property.

\begin{thm}[Subaddivity of Matrix $\phi$-Entropy II]
\label{thm:algebra-subadd}
Fix a function $\phi \in \Phi_{\infty}$.  Let $\mathfrak{A}_1, \dots, \mathfrak{A}_n$ be $*$-subalgebras of $\M^d$ whose conditional expectations commute.  Then
$$
H_{\phi}(\mtx{A} \condl \mathfrak{A}_1, \dots, \mathfrak{A}_n)
	\leq \sum\nolimits_{i=1}^n H_{\phi}(\mtx{A} \condl \mathfrak{A}_i)
	\quad\text{for $\mtx{A} \in \Sym_+^d$.}
$$
\end{thm}

\noindent
We omit the proof of this result.  The argument involves considerations similar with Theorem~\ref{thm:subadditivity-intro}, but it requires an extra dose of operator theory.  The work in this paper already addresses the more challenging aspects of the proof.

Theorem~\ref{thm:algebra-subadd} can be seen as a formal extension of the subadditivity of matrix $\phi$-entropy expressed in Theorem~\ref{thm:subadditivity-intro}.  To see why, let $\Omega:=\Omega_1 \times \dots \times \Omega_n$ be a product probability space.  The space $L_2( \Omega; \M^d )$ of random matrices is a $*$-algebra with the normalized trace functional $\Expect \ntr$.  For each $i = 1, \dots, n$, we can form a $*$-subalgebra $\mathfrak{A}_i$ consisting of the random matrices that do not depend on the $i$th factor $\Omega_i$ of the product.  The conditional expectation $\Expect_{\mathfrak{A}_i}$ simply integrates out the $i$th random variable.  By independence, the family of conditional expectations $\Expect_{\mathfrak{A}_1}, \dots, \Expect_{\mathfrak{A}_n}$ commutes.  Using this dictionary, compare the statement of Theorem~\ref{thm:algebra-subadd} with Theorem~\ref{thm:subadditivity-intro}.

\subsection{Background on the Entropy Method}

This section contains a short summary of related work on the entropy method and on matrix concentration.

Inspired by Talagrand's work~\cite{Talagrand1991} on concentration in product spaces,
Ledoux~\cite{Ledoux1996,Ledoux2001} and Bobkov \& Ledoux~\cite{Bobkov1998} developed
the entropy method for obtaining concentration inequalities on product spaces.
This approach is based on new logarithmic Sobolev inequalities for product spaces.
Other authors, including Massart~\cite{Massart2000,Massart2000a}, Rio~\cite{Rio2001}, Bousquet~\cite{Bousquet2002}, and Boucheron et al.~\cite{Boucheron2003} have extended these results to obtain additional concentration inequalities.  See the book~\cite{BLM13:Concentration-Inequalities} for a comprehensive treatment.

In a related line of work, Lata{\l}a \& Oleskiewicz~\cite{Latala2000} and Chafa{\"i}~\cite{Chafai2004} investigated generalizations of the logarithmic Sobolev inequalities based on $\phi$-entropy functionals.  The paper~\cite{Boucheron2005} of Boucheron et al.~elaborates on these ideas to obtain a new class of moment inequalities; see also the book~\cite{BLM13:Concentration-Inequalities}.  Our paper is based heavily on the approach in~\cite{Boucheron2005}.

There is a recent line of work that develops concentration inequalities for random matrices by adapting classical arguments from the theory of concentration of measure.  The introduction contains an overview of this research, so we will not repeat ourselves.

Our paper can be viewed as a first attempt to obtain concentration inequalities for random matrices using the entropy method.  In spirit, our approach is closely related to arguments based on Stein's method~\cite{Mackey2012,Paulin2013}.  The theory in our paper does not require any notion of exchangeable pairs.  On the other hand, the arguments here are substantially harder, and they still result in weaker concentration bounds.

For the classical entropy $\phi : t \mapsto t\log t$, we are aware of some precedents for our subadditivity results, Theorem~\ref{thm:subadditivity-intro} and Theorem~\ref{thm:algebra-subadd}.  In this special case, the subadditivity of $H_{\phi}$ already follows from a classical result~\cite{Lindblad1973}.  There is a more modern paper~\cite{HOZ01:Upper-Bound} that contains a similar type of subadditivity bound.  Very recently, Hansen~\cite{Han13:Convexity-Residual} has identified a convexity property of another related quantity, called the residual entropy.  Note that, in the literature on quantum statistical mechanics and quantum information theory, the phrase ``subadditivity of entropy'' refers to a somewhat different kind of bound~\cite{LR73:Proof-Strong}; see~\cite[Sec.~8]{Carlen2010} for a modern formulation.

\subsection{Roadmap} %

Section~\ref{section:preliminaries} contains some background on operator theory.  In Section~\ref{section:subadditivity}, we prove Theorem~\ref{thm:subadditivity-intro} on the subadditivity on the matrix $\phi$-entropy.  Section~\ref{section:exchangeable} describes how to obtain Corollary~\ref{cor:subadditivity-intro}.  Afterward, in Section~\ref{section:member}, prove that the standard entropy and certain power functions belong to the $\Phi_{\infty}$ function class.  Finally, Sections~\ref{section:concentration} and~\ref{section:moment} derive the matrix concentration inequalities, Theorem~\ref{thm:tail-bound} and Theorem~\ref{thm:moment-bound-intro}.

\section{Operators and Functions acting on Matrices} \label{section:preliminaries}

This work involves a substantial amount of operator theory.  This section contains a short treatment of the basic facts.  See~\cite{Bhatia1997,Bhatia2007} for a more complete introduction.

\subsection{Linear Operators on Matrices} \label{subsection:operator}

Let $\C^d$ be the complex Hilbert space of dimension $d$, equipped with the standard inner product $\ip{ \vct{a} }{ \vct{b} }: = \vct{a}^\adj \vct{b}$.  We usually identify $\M^d$ with $\B(\C^d)$, the complex Banach space of linear operators acting on $\C^d$, equipped with the $\ell_2$ operator norm $\norm{\cdot}$.

We can also endow $\M^d$ with the normalized trace inner product $\ip{ \mtx{A} }{ \mtx{B} } := \ntr( \mtx{A}^* \mtx{B})$ to form a Hilbert space.  As a Hilbert space, $\M^d$ is isometrically isomorphic with $\C^{d^2}$.
Let $\B(\M^d)$ denote the complex Banach space of linear operators that map the Hilbert space $\M^d$ into itself, equipped with the induced operator norm.  The Banach space $\B(\M^d)$ is isometrically isomorphic with the Banach space $\M^{d^2}$.

As a consequence of this construction, every concept from matrix analysis has an immediate analog for linear operators on matrices.  An operator $\op{T} \in \B(\M^d)$ is \term{self-adjoint} when
$$
\ip{ \mtx{A} }{ \op{T}(\mtx{B}) } =
\ip{ \op{T}(\mtx{A}) }{ \mtx{B} }
\quad\text{for all $\mtx{A}, \mtx{B} \in \B(\M^d)$.}
$$
A self-adjoint operator $\op{T} \in \B(\M^d)$ is \term{positive semidefinite} when
$$
\ip{ \mtx{A} }{ \op{T}(\mtx{A}) } \geq 0
\quad\text{for all $\mtx{A} \in \M^d$.}
$$
For self-adjoint operators $\op{S}, \op{T} \in \B(\M^d)$, the notation $\op{S} \psdle \op{T}$ means that $\op{T} - \op{S}$ is positive semidefinite.

Each self-adjoint matrix operator $\op{T} \in \B(\M^d)$ has a spectral resolution of the form
\begin{equation} \label{eqn:op-resolution}
\op{T} = \sum\nolimits_{i=1}^{d^2} \lambda_i \op{P}_i
\end{equation}
where $\lambda_1, \dots, \lambda_{d^2}$ are the eigenvalues of $\op{T}$ and the spectral projectors $\op{P}_1, \dots, \op{P}_{d^2}$ are positive-semidefinite operators that satisfy
$$
\op{P}_i \op{P}_j = \delta_{ij} \op{P}_{j}
\quad\text{and}\quad
\sum\nolimits_{i=1}^{d^2} \op{P}_i = \op{I},
$$
where $\delta_{ij}$ is the Kronecker delta and $\op{I}$ is the identity operator.  As in the matrix case, a self-adjoint operator with nonnegative eigenvalues is the same thing as a positive-semidefinite operator.

We can extend a scalar function $f : I \to \R$ on an interval $I$ of the real line to obtain a standard operator function.  Indeed, if $\op{T}$ has the spectral resolution~\eqref{eqn:op-resolution} and the eigenvalues of $\op{T}$ fall in the interval $I$, we define
$$
f(\op{T}) := \sum\nolimits_{i=1}^{d^2} f(\lambda_i) \op{P}_i.
$$
This definition, of course, parallels the definition for matrices.

\subsection{Monotonicity and Convexity}

Let $X$ and $Y$ be sets of self-adjoint operators, such as $\Sym^d(I)$ or the set of self-adjoint operators in $\B(\M^d)$.  We can introduce notions of monotonicity and convexity for a general function $\op{\Psi} : X \to Y$ using the semidefinite order on the spaces of operators.

\begin{defn}[Monotone Operator-Valued Function]
The function $\op{\Psi} : X \to Y$ is \term{monotone} when
$$
\op{S} \psdle \op{T}
\quad\Longrightarrow\quad
\op{\Psi}(\op{S}) \psdle \op{\Psi}(\op{T})
\quad\text{for all $\op{S}, \op{T} \in X$.}
$$
\end{defn}

\begin{defn}[Convex Operator-Valued Function]
The function $\op{\Psi} : X \to Y$ is \term{convex} when $X$ is a convex set and
$$
\op{\Psi}( \alpha \op{S} + \bar{\alpha} \op{T} )
	\psdle \alpha \cdot \op{\Psi}(\op{S}) + \bar{\alpha} \cdot \op{\Psi}(\op{T})
	\quad\text{for all $\alpha \in [0,1]$ and all $\op{S},\op{T} \in X$.}
$$
We have written $\bar{\alpha} := 1 - \alpha$.  The function $\op{\Psi}$ is \term{concave} when $-\op{\Psi}$ is convex.
\end{defn}

\noindent
The convexity of an operator-valued function $\op{\Psi}$ is equivalent with a
Jensen-type relation:
\begin{equation} \label{eqn:op-jensen}
\op{\Psi}(\Expect \op{X}) \psdle \Expect \op{\Psi}(\op{X})
\end{equation}
whenever $\op{X}$ is an integrable random operator taking values in $X$.

In particular, we can apply these definitions to standard matrix and operator functions.  Let $I$ be an interval of the real line.  We say that the function $f : I \to \R$ is \term{operator monotone} when the lifted map $f : \Sym^d(I) \to \Sym^d$ is monotone for each natural number $d$.  Likewise, the function $f : I \to \R$ is \term{operator convex} when the lifted map $f : \Sym^d(I) \to \Sym^d$ is convex for each natural number $d$.

Although scalar monotonicity and convexity are quite common, they are much rarer in the matrix setting~\cite[Chap.~4]{Bhatia1997}.  For present purposes, we note that the power functions $t \mapsto t^p$ with $p \in [0,1]$ are operator monotone and operator concave.  The power functions $t \mapsto t^{p}$ with $p \in [1, 2]$ and the standard entropy $t \mapsto t \log t$ are all operator convex.

\subsection{The Derivative of a Vector-Valued Function}

The definition of the $\Phi_{\infty}$ function class involves a requirement that a certain standard matrix function is differentiable.  For completeness, we include the background needed to interpret this condition.

\begin{defn}[Derivative of a Vector-Valued Function] \label{defn:frechet}
Let $X$ and $Y$ be Banach spaces, and let $U$ be an open subset of $X$.  
A function $\mtx{F} : U \to Y$ is \term{differentiable} at a point $\mtx{A} \in U$ if there exists a bounded linear operator $\op{T} : X \to Y$ for which
\begin{equation*}
\lim_{\mtx{B} \to \zeromtx} \frac{\norm{\mtx{F}(\mtx{A} + \mtx{B}) - \mtx{F}(\mtx{A}) - \op{T}(\mtx{B})}_Y}{\norm{\mtx{B}}_X} = 0.
\end{equation*}
When $\mtx{F}$ is differentiable at $\mtx{A}$, the operator $\mathsf{T}$ is called the \term{derivative} of $\mtx{F}$ at $\mtx{A}$, and we define $\op{D}\mtx{F}(\mtx{A}) := \op{T}$.
\end{defn}

\noindent
The derivative and the directional derivative have the following relationship:
\begin{equation} \label{eqn:D-dif}
\frac{\diff{}}{\diff{s}}\mtx{F}(\mtx{A} + s \mtx{B}) \Big|_{s=0} = \op{D}\mtx{F}(\mtx{A}) (\mtx{B}).
\end{equation}
In Section~\ref{subsection:derivative}, we present an explicit formula for the derivative of a standard matrix function.

\section{Subadditivity of Matrix \texorpdfstring{$\phi$}{phi}-Entropy} \label{section:subadditivity}

In this section, we establish Theorem~\ref{thm:subadditivity-intro}, which states that the matrix $\phi$-entropy is subadditive for every function in the $\Phi_\infty$ class.
This result depends on a variational representation for the matrix $\phi$-entropy that appears in Section~\ref{subsection:lower}.  We use the variational formula to derive a Jensen-type inequality in Section~\ref{subsection:jensen}.  The proof of Theorem~\ref{thm:subadditivity-intro} appears in Section~\ref{subsection:proof}.

\subsection{Representation of Matrix $\phi$-Entropy as a Supremum} \label{subsection:lower}

The fundamental fact behind the subadditivity theorem is a representation of the matrix $\phi$-entropy as a supremum of affine functions.

\begin{lemma}[Supremum Representation for Entropy]
\label{lem:lower-variation}  %
Fix a function $\phi \in \Phi_\infty$, and introduce the scalar derivative $\psi = \phi'$.
Suppose that $\mtx{Z}$ is a random positive-semidefinite matrix for which $\norm{\mtx{Z}}$ and $\norm{\phi(\mtx{Z})}$ are integrable.  Then
\begin{equation} \label{eqn:lower-variation}
H_\phi ( \mtx{Z})  = \sup_{\mtx{T}} \ \Expect \ntr \big[ (\psi(\mtx{T}) - \psi(\Expect \mtx{T}))(\mtx{Z} - \mtx{T}) + \phi(\mtx{T}) - \phi(\Expect \mtx{T}) \big].
\end{equation}
The range of the supremum contains each random positive-definite matrix $\mtx{T}$ for which $\norm{\mtx{T}}$ and $\norm{\phi(\mtx{T})}$ are integrable.  In particular, the matrix $\phi$-entropy $H_\phi$ can be written in the dual form
\begin{equation} \label{eqn:duality}
H_\phi(\mtx{Z}) = \sup_{\mtx{T}}\ \Expect \ntr \big[ \mtx{\Upsilon}_1(\mtx{T}) \cdot \mtx{Z} + \mtx{\Upsilon}_2(\mtx{T}) \big],
\end{equation}
where $\mtx{\Upsilon}_i : \Sym^d_+ \to \Sym^d$ for $i = 1, 2$.
\end{lemma}

\noindent
This result implies that $H_{\phi}$ is a convex function on the space of random positive-semidefinite matrices.  The dual representation of $H_{\phi}$ is well suited for establishing a form of Jensen's inequality, Lemma~\ref{lem:jensen}, which is the main ingredient in the proof of the subadditivity property, Theorem~\ref{thm:subadditivity-intro}.

It may be valuable to see some particular instances of the dual representation of the matrix $\phi$-entropy:
\begin{align*}
&H_{\phi}(\mtx{Z}) = \sup_{\mtx{T}}\ \Expect \ntr \big[ (\log \mtx{T} - \log (\Expect \mtx{T})) \cdot \mtx{Z} \big] &&\text{when $\phi(t) = t \log t$.} \\
&H_{\phi}(\mtx{Z}) =\sup_{\mtx{T}} \ \Expect \ntr \big[ p(\mtx{T}^{p-1} - (\Expect \mtx{T})^{p-1}) \cdot \mtx{Z} - (p-1)(\mtx{T}^p - (\Expect \mtx{T})^p)  \big] &&\text{when $\phi(t) = t^p$ for $p \in [1,2]$.} %
\end{align*}
The first formula is the matrix version of a well-known variational principle for the classical entropy, cf.~\cite[p.~525]{Boucheron2005}.  In the matrix setting, this result can be derived from the joint convexity of quantum relative entropy~\cite{Lindblad1973}.

\subsubsection{The Convexity Lemma} \label{subsubsection:convexity}

To establish the variational formula, we require a convexity result for a quadratic form connected with the function $\phi$.

\begin{lemma} \label{lem:technical}
Fix a function $\phi \in \Phi_{\infty}$, and let $\psi = \phi'$.  Suppose that $\mtx{Y}$ is a random matrix taking values in $\Sym^d_+$, and let $\mtx{K}$ be a random matrix taking values in $\M^d$.  Assume that $\norm{\mtx{Y}}$ and $\norm{\mtx{K}}$ are integrable.  Then
$$
\Expect \ip{ \mtx{K} }{ \op{D} \psi(\mtx{Y})(\mtx{K}) }
	\geq \ip{ (\Expect \mtx{K}) }{ \op{D} \psi(\Expect \mtx{Y})(\Expect \mtx{K}) }
$$
\end{lemma}

\begin{proof}
The proof hinges on a basic convexity property of quadratic forms.  Define a map that takes a matrix $\mtx{A}$ in $\Sym^d$ and a positive-definite operator $\op{T}$ on $\M^d$ to a nonnegative number:
$$
\mathscr{Q} : (\mtx{A}, \op{T}) \mapsto \ip{ \mtx{A} }{ \op{T}^{-1}(\mtx{A}) }.
$$
We assert that the function $\mathscr{Q}$ is convex.  Indeed, the same result is well known when $\mtx{A}$ and $\op{T}$ are replaced by a vector and a positive-definite matrix~\cite[Exer.~1.5.1]{Bhatia2007}, and the extension is immediate from the isometric isomorphism between operators and matrices.

Recall that the $\Phi_{\infty}$ class requires $\mtx{A} \mapsto [\op{D} \psi(\mtx{A})]^{-1}$ to be a concave map on $\Sym^d_{++}$.  With these observations at hand, we can make the following calculation:
\begin{align*}
\Expect \ip{ \mtx{K} }{ \op{D}\psi(\mtx{Y})(\mtx{K}) }
&= \Expect \ip{ \mtx{K} }{ ([\op{D}\psi(\mtx{Y})]^{-1})^{-1} (\mtx{K}) } \\
&\geq \ip{ (\Expect \mtx{K}) }{ (\Expect [ \op{D} \psi(\mtx{Y})]^{-1} )^{-1} (\Expect \mtx{K}) } \\
& \geq \ip{ (\Expect \mtx{K}) }{ ([ \op{D}\psi(\Expect \mtx{Y})]^{-1} )^{-1}(\Expect \mtx{K}) } \\
& = \ip{ (\Expect \mtx{K}) }{ \op{D}\psi(\Expect \mtx{Y})(\Expect \mtx{K}) }.
\end{align*}
We obtain the second relation when we apply Jensen's inequality to the convex function $\mathscr{Q}$.  The third relation depends on the semidefinite Jensen inequality~\eqref{eqn:op-jensen} for the concave function $\mtx{A} \mapsto [\op{D}\psi(\mtx{A})]^{-1}$, coupled with the fact~\cite[Prop.~V.1.6]{Bhatia1997} that the operator inverse reverses the semidefinite order.
\end{proof}

\subsubsection{Proof of Lemma \ref{lem:lower-variation}} \label{subsubsection:proof-lower}

The argument parallels the proof of~\cite[Lem.~1]{Boucheron2005}.  We begin with some reductions.  The case where $\phi$ is an affine function is immediate, so we may require the derivative $\psi = \phi'$ to be non-constant.  By approximation, we may also assume that the random matrix $\mtx{Z}$ is strictly positive definite.

[Indeed, since $\phi$ is continuous on $\R_+$, the Dominated Convergence Theorem implies that the matrix $\phi$-entropy $H_{\phi}$ is continuous on the set containing each positive-semidefinite random matrix $\mtx{Y}$ where $\norm{ \mtx{Y} }$ and $\norm{\phi(\mtx{Y})}$ are integrable.  Therefore, we can approximate a positive-semidefinite random matrix $\mtx{Z}$ by a sequence $\{\mtx{Y}_n\}$ of positive-definite random matrices where $\mtx{Y}_n \to \mtx{Z}$ and be confident that $H_{\phi}(\mtx{Y}_n) \to H_{\phi}(\mtx{Z})$.]

When $\mtx{T} = \mtx{Z}$, the argument of the supremum in~\eqref{eqn:lower-variation} equals $H_{\phi}(\mtx{Z})$.  Therefore, our burden is to verify the inequality
\begin{equation} \label{eqn:ineq}
H_{\phi}(\mtx{Z}) \geq \Expect \ntr \big[ (\psi(\mtx{T}) - \psi(\Expect \mtx{T}))(\mtx{Z} - \mtx{T}) + \Expect \phi(\mtx{T}) - \phi(\Expect \mtx{T}) \big]
\end{equation}
for each random positive-definite matrix $\mtx{T}$ that satisfies the same integrability requirements as $\mtx{Z}$.
For simplicity, we assume that the eigenvalues of both $\mtx{Z}$ and $\mtx{T}$ are bounded and bounded away from zero.  See Appendix~\ref{app:general} for the extension to the general case.

We use an interpolation argument to establish \eqref{eqn:ineq}.  Define the family of random matrices
\begin{equation*}\label{eqn:T-s}
\mtx{T}_s := (1 - s) \cdot \mtx{Z} + s \cdot \mtx{T}\quad \text{for $s \in [0, 1]$}.
\end{equation*}
Introduce the real-valued function
\begin{equation*}
F(s) := \Expect \ntr \big[ (\psi(\mtx{T}_s) - \psi(\Expect \mtx{T}_s) )\cdot (\mtx{Z} - \mtx{T}_s) \big] + H_\phi( \mtx{T}_s).
\end{equation*}
Observe that $F(0) = H_{\phi}(\mtx{Z})$, while $F(1)$ coincides with the right-hand side of~\eqref{eqn:ineq}.  Therefore, to establish~\eqref{eqn:ineq}, it suffices to show that the function $F(s)$ is weakly decreasing on the interval $[0,1]$.

We intend to prove that $F'(s) \leq 0$ for $s \in [0, 1]$.  Since $\mtx{Z} - \mtx{T}_s = -s \cdot (\mtx{T} - \mtx{Z})$, we can rewrite the function $F$ in the form
\begin{equation} \label{eqn:F-convenient}
F(s) = -s  \cdot\Expect \ntr \big[ (\psi(\mtx{T}_s) - \psi(\Expect \mtx{T}_s)) \cdot (\mtx{T} - \mtx{Z})\big] + \Expect \ntr \big[ \phi(\mtx{T}_s) - \phi(\Expect \mtx{T}_s)) \big].
\end{equation}
We differentiate the function $F$ to obtain
\begin{multline} \label{eqn:F'-1}
F'(s) = -s \cdot \Expect \ntr \big[ \op{D}\psi(\mtx{T}_s) (\mtx{T} - \mtx{Z}) \cdot (\mtx{T} - \mtx{Z}) \big]
	+ s \cdot \ntr \big[  \op{D}\psi(\mtx{T}_s) (\Expect(\mtx{T} - \mtx{Z})) \cdot (\Expect (\mtx{T} - \mtx{Z})) \big]\\
	- \Expect \ntr \big[ (\psi(\mtx{T}_s) - \psi(\Expect \mtx{T}_s)) \cdot  (\mtx{T} - \mtx{Z})\big]
	+ \Expect \ntr \big[ (\psi(\mtx{T}_s) - \psi(\Expect \mtx{T}_s)) \cdot (\mtx{T} - \mtx{Z}) \big].
\end{multline}
To handle the first term in~\eqref{eqn:F-convenient}, we applied the product rule, the rule~\eqref{eqn:D-dif} for directional derivatives, and the expression $\diff{}\mtx{T}_s / \diff{s} = \mtx{T} - \mtx{Z}$.  We used the identity
$\op{D} \trace \phi(\mtx{A}) = \psi(\mtx{A})$
to differentiate the second term.
We also relied on the Dominated Convergence Theorem to pass derivatives through  expectations, which is justified because $\phi$ and $\psi$ are continuously differentiable on $\Sym^d_{++}$ and the eigenvalues of the random matrices are bounded and bounded away from zero.  Now, the last two terms in~\eqref{eqn:F'-1} cancel, and we can rewrite the first two terms using the trace inner product:
$$
F'(s) = s \cdot \big[ \ip{ (\Expect(\mtx{T}-\mtx{Z})) }{ \op{D} \psi(\Expect \mtx{T}_s) (\Expect(\mtx{T}-\mtx{Z})) } -
\Expect \ip{ (\mtx{T} - \mtx{Z}) }{ \op{D}\psi(\mtx{T}_s)(\mtx{T} - \mtx{Z}) } \big].
$$
Invoke Lemma~\ref{lem:technical} to conclude that $F'(s) \leq 0$ for $s \in [0,1]$.

\subsection{A Conditional Jensen Inequality} \label{subsection:jensen}

The variational inequality in Lemma~\ref{lem:lower-variation} leads directly to a Jensen inequality for the matrix $\phi$-entropy.

\begin{lemma}[Conditional Jensen Inequality] \label{lem:jensen}
Fix a function $\phi \in \Phi_\infty$.  Suppose that $(X_1, X_2)$ is a pair of independent random variables taking values in a Polish space, and let $\mtx{Z} = \mtx{Z}(X_1, X_2)$ be a random positive-semidefinite matrix for which $\norm{\mtx{Z}}$ and $\norm{\phi(\mtx{Z})}$ are integrable.  Then
\begin{equation*} \label{eqn:jensen}
H_\phi \left( \Expect_1 \mtx{Z} \right) \leq \Expect_1 H_\phi \left( \mtx{Z} \condl X_1\right),
\end{equation*}
where $\Expect_1$ is the expectation with respect to the first variable $X_1$.
\end{lemma}

\begin{proof}
Let $\Expect_2$ denote the expectation with respect to the second variable $X_2$.
The result is a simple consequence of the dual representation~\eqref{eqn:duality} of the matrix $\phi$-entropy:
\begin{equation}
H_\phi \left( \Expect_1 \mtx{Z}  \right)
	= \sup_{\mtx{T}}\ \Expect_2 \ntr \left[ \mtx{\Upsilon}_1\big(\mtx{T}(X_2)\big)
	\cdot (\Expect_1 \mtx{Z})  + \mtx{\Upsilon}_2 \big(\mtx{T}(X_2)\big) \right]. \label{eqn:rewrite}
\end{equation}
We have written $\mtx{T}(X_2)$ to emphasize that this matrix depends only on the randomness in $X_2$.  To control \eqref{eqn:rewrite}, we apply Fubini's theorem to interchange the order of $\Expect_1$ and $\Expect_2$, and then we exploit the convexity of the supremum to draw out the expectation $\Expect_1$.  
\begin{align*}
H_\phi \left( \Expect_1 \mtx{Z} \right)
&= \sup_{\mtx{T}} \ \Expect_1 \Expect_2 \ntr \left[
	\mtx{\Upsilon}_1(\mtx{T}(X_2)) \cdot \mtx{Z}
	+ \mtx{\Upsilon}_2(\mtx{T}(X_2))  \right]  \\
&\leq \Expect_1 \sup_{\mtx{T}} \ \Expect_2 \ntr
	\left[ \mtx{\Upsilon}_1(\mtx{T}(X_2)) \cdot \mtx{Z}
	+ \mtx{\Upsilon}_2(\mtx{T}(X_2) \right]  \\
& = \Expect_1 \sup_{\mtx{T}} \ \Expect \big[ \ntr [
	\mtx{\Upsilon}_1(\mtx{T}(X_2)) \cdot \mtx{Z}
	+ \mtx{\Upsilon}_2(\mtx{T}(X_2) ] \condl X_1  \big] \\
& = \Expect_1 H_\phi ( \mtx{Z} \condl X_1 ).
\end{align*}
The last relation is the duality formula \eqref{eqn:duality}, applied conditionally.
\end{proof}

\subsection{Proof of Theorem \ref{thm:subadditivity-intro}} \label{subsection:proof}

We are now prepared to establish the main result on subadditivity of matrix $\phi$-entropy.  This theorem is a direct consequence of the conditional Jensen inequality, Lemma~\ref{lem:jensen}.  In this argument, we write $\Expect_i$ for the expectation with respect to the variable $X_i$.  Using the notation from Section~\ref{subsection:subadditivity}, we see that $\Expect_i  = \Expect[ \, \cdot \condl \vct{x}_{-i} ]$.

First, separate the matrix $\phi$-entropy into two parts by adding and subtracting terms:
\begin{align}
H_\phi( \mtx{Z})
&= \Expect \ntr{}[ \phi(\mtx{Z})
	- \phi( \Expect_1 \mtx{Z}) + \phi(\Expect_1 \mtx{Z})
	- \phi( \Expect \mtx{Z} ) ].
	\notag \\
&= \Expect \big[ \Expect_1 \ntr{}[ \phi(\mtx{Z})
	- \phi( \Expect_1 \mtx{Z})  ] \big]
+  \Expect \ntr{}[ \phi(\Expect_1 \mtx{Z})
	- \phi( \Expect \Expect_1 \mtx{Z} ) ].
	\label{eqn:first-term}
\end{align}
We can rewrite this expression as
\begin{align} \label{eqn:separate}
H_{\phi}(\mtx{Z})
	&= \Expect H_\phi (\mtx{Z} \condl \vct{x}_{-1} )
	+ H_\phi ( \Expect_1 \mtx{Z} ) \notag \\
	&\leq \Expect H_\phi (\mtx{Z} \condl \vct{x}_{-1} )
	+ \Expect_1 H_\phi (\mtx{Z} \condl X_1 ).
\end{align}
The inequality follows from Lemma~\ref{eqn:jensen} because $\mtx{Z} = \mtx{Z}(X_1, \vct{x}_{-1})$ where $X_1$ and $\vct{x}_{-1}$ are independent random variables.

The first term on the right-hand side of~\eqref{eqn:separate} coincides with the first summand on the right-hand side of the subadditivity inequality~\eqref{eqn:subadditivity-intro}.  We must argue that the remaining summands are contained in the second term on the right-hand side of~\eqref{eqn:separate}.  Repeating the argument in the previous paragraph, conditioning on $X_1$, we obtain 
$$
H_\phi (\mtx{Z} \condl X_1 )
	\leq \Expect\big[ H_{\phi}(\mtx{Z} \condl \vct{x}_{-2}) \condl X_1 ]
	+ \Expect_2  H_{\phi}(\mtx{Z} \condl X_1, X_2).
$$
Substituting this expression into~\eqref{eqn:separate}, we obtain
$$
H_{\phi}(\mtx{Z})
	\leq \sum\nolimits_{i=1}^2 \Expect H_\phi (\mtx{Z} \condl \vct{x}_{-i} )
	+ \Expect_1 \Expect_2 H_{\phi}(\mtx{Z} \condl X_1, X_2).
$$
Continuing in this fashion, we arrive at the subadditivity inequality \eqref{eqn:subadditivity-intro}:
$$
H_{\phi}(\mtx{Z})
	\leq \sum\nolimits_{i=1}^n \Expect H_\phi (\mtx{Z} \condl \vct{x}_{-i} ).
$$
This completes the proof of Theorem~\ref{thm:subadditivity-intro}.

\section{Entropy Bounds via Exchangeability}
\label{section:exchangeable}

In this section, we derive Corollary~\ref{cor:subadditivity-intro}, which uses exchangeable pairs to bound the conditional entropies that appear in Theorem~\ref{thm:subadditivity-intro}.  This result follows from another variational representation of the matrix $\phi$-entropy. %

\subsection{Representation of the Matrix $\phi$-Entropy as an Infimum} \label{subsection:upper}

In this section, we present another formula for the matrix $\phi$-entropy.

\begin{lemma}[Infimum Representation for Entropy] \label{lem:entropy-bound}
Fix a function $\phi \in \Phi_{\infty}$, and let $\psi = \phi'$.
Assume that $\mtx{Z}$ is a random positive-semidefinite matrix where
$\norm{\mtx{Z}}$ and $\norm{\phi(\mtx{Z})}$ are integrable.  Then
\begin{equation} \label{eqn:entropy-bound}
H_\phi ( \mtx{Z})  = \inf_{\mtx{A} \in \Sym^d_+} \Expect \ntr
\left[ \phi(\mtx{Z}) - \phi(\mtx{A}) - (\mtx{Z} - \mtx{A}) \cdot \psi(\mtx{A}) \right].
\end{equation}
Let $\mtx{Z}'$ be an independent copy of $\mtx{Z}$.  Then
\begin{equation} \label{eqn:symmetry}
H_\phi(\mtx{Z}) \leq \frac{1}{2} \cdot \Expect \ntr \left[ (\mtx{Z} - \mtx{Z}') (\psi(\mtx{Z}) - \psi(\mtx{Z}'))   \right].
\end{equation}
\end{lemma}

We require a familiar trace inequality~\cite[Thm.~2.11]{Carlen2010}.  This bound simply restates the fact that a convex function lies above its tangents.

\begin{prop}[Klein's Inequality] \label{prop:klein}
Let $f : I \to \R$ be a differentiable convex function on an interval $I$ of the real line.  Then
\begin{equation*}
\ntr \big[f(\mtx{B}) - f(\mtx{A}) - (\mtx{B}-\mtx{A}) \cdot f'(\mtx{A})] \geq 0
\quad\text{for all $\mtx{A}, \mtx{B} \in \Sym^d(I)$.}
\end{equation*}
\end{prop}

\noindent
With Klein's inequality at hand, the variational inequality follows quickly.

\begin{proof}[Proof of Lemma~\ref{lem:entropy-bound}]
Every function $\phi \in \Phi_{\infty}$ is convex and differentiable, so Proposition~\ref{prop:klein} with $\mtx{B} = \Expect \mtx{Z}$ implies that
\begin{equation*}
\ntr{} [ -  \phi(\Expect \mtx{Z}) ]
	\leq \ntr{} [ - \phi(\mtx{A}) - (\Expect \mtx{Z} - \mtx{A}) \cdot \psi(\mtx{A}) ]
\end{equation*}
for each fixed matrix $\mtx{A} \in \Sym_+^d$.  Substitute this bound into the definition~\eqref{eqn:entropy-intro} of the matrix $\phi$-entropy, and draw the expectation out of the trace to reach
\begin{equation} \label{eqn:H-inequality}
H_\phi(\mtx{Z}) \leq  \Expect  \ntr{} [\phi(\mtx{Z}) - \phi(\mtx{A}) - (\mtx{Z} - \mtx{A})\cdot \psi(\mtx{A}) ].
\end{equation}
The inequality~\eqref{eqn:H-inequality} becomes an equality when $\mtx{A} = \Expect \mtx{Z}$, which establishes the variational representation~\eqref{eqn:entropy-bound}. 

The symmetrized bound~\eqref{eqn:symmetry} follows from an exchangeability argument.  Select $\mtx{A} = \mtx{Z}'$ in the expression~\eqref{eqn:H-inequality}, and apply the fact that $\Expect \phi(\mtx{Z}) = \Expect \phi(\mtx{Z}')$ to obtain
\begin{align}
H_\phi(\mtx{Z}) \leq - \Expect  \ntr{} [ (\mtx{Z} - \mtx{Z}') \cdot \psi(\mtx{Z}') ]. \label{eqn:bound1}
\end{align}
Since $\mtx{Z}$ and $\mtx{Z}'$ are exchangeable, we can also bound the matrix $\phi$-entropy as
\begin{equation} \label{eqn:bound2}
H_{\phi}(\mtx{Z}) \leq  - \Expect \ntr{}[ (\mtx{Z}' - \mtx{Z}) \cdot \psi(\mtx{Z}) ].
\end{equation}
Take the average of the two bounds \eqref{eqn:bound1} and \eqref{eqn:bound2} to reach the desired inequality \eqref{eqn:symmetry}.
\end{proof}

\subsection{Proof of Corollary~\ref{cor:subadditivity-intro}}

Lemma \ref{lem:entropy-bound} leads to a succinct proof of Corollary \ref{cor:subadditivity-intro}.  We continue to use the notation from
Section~\ref{subsection:subadditivity}.
Apply the inequality~\eqref{eqn:symmetry}
conditionally to control the conditional matrix $\phi$-entropy:
\begin{equation} \label{eqn:symmetry-con}
H_{\phi}(\mtx{Z} \condl \vct{x}_{-i})
\leq \frac{1}{2} \cdot \Expect \ntr \big[ (\mtx{Z} - \mtx{Z}'_i)  (\psi(\mtx{Z}) - \psi(\mtx{Z}'_i)) \condl \vct{x}_{-i}\big]
\end{equation}
because $\mtx{Z}'_i$ and $\mtx{Z}$ are conditionally iid, given $\vct{x}_{-i}$.
Take the expectation on both sides of \eqref{eqn:symmetry-con},
and invoke the tower property of conditional expectation:
\begin{equation} \label{eqn:symmetry-con2}
\Expect H_{\phi}(\mtx{Z} \condl \vct{x}_{-i}) \leq \frac{1}{2} \cdot \Expect \ntr
\big[ (\mtx{Z} - \mtx{Z}'_i) (\psi(\mtx{Z}) - \psi(\mtx{Z}'_i)) \big].
\end{equation}
To complete the proof, substitute~\eqref{eqn:symmetry-con2} into the right-hand side of
the bound~\eqref{eqn:subadditivity-intro} from the subadditivity
result, Theorem~\ref{thm:subadditivity-intro}.

\section{Members of the $\Phi_\infty$ function class}\label{section:member}

In this section, we demonstrate that the classical entropy and certain power
functions belong to the $\Phi_{\infty}$ function class.  The main challenge
is to verify that $\mtx{A} \mapsto [\op{D} \psi(\mtx{A})]^{-1}$
is a concave operator-valued map.  We establish this result for
the classical entropy in Section~\ref{subsection:regular-entropy}
and for the power function in Section~\ref{subsection:power}.

\subsection{Tensor Product Operators}

First, we explain the tensor product construction of
an operator.  The tensor product will allow us to represent
the derivative of a standard matrix function compactly.

\begin{defn}[Tensor Product]
Let $\mtx{A}, \mtx{B} \in \Sym^d$.  The operator $\mtx{A} \otimes \mtx{B} \in \B(\M^d)$ is defined by the relation
\begin{equation} \label{eqn:tensor}
(\mtx{A} \otimes \mtx{B})( \mtx{M} ) = \mtx{AMB} \quad \text{for each $\mtx{M} \in \M^d$}.
\end{equation}
The operator $\mtx{A} \otimes \mtx{B}$ is self-adjoint because we assume the factors are Hermitian matrices.
\end{defn}

Suppose that $\mtx{A}, \mtx{B} \in \Sym^d$ are Hermitian matrices with spectral resolutions
\begin{equation} \label{eqn:AB-res}
\mtx{A} = \sum\nolimits_{i=1}^d \lambda_i \mtx{P}_i
\quad\text{and}\quad
\mtx{B} = \sum\nolimits_{j=1}^d \mu_j \mtx{Q}_j.
\end{equation}
Then the tensor product $\mtx{A} \otimes \mtx{B}$ has the spectral resolution
$$
\mtx{A} \otimes \mtx{B} = \sum\nolimits_{i,j=1}^d \lambda_i \mu_j \mtx{P}_i \otimes \mtx{Q}_j.
$$
In particular, the tensor product of two positive-definite matrices is a positive-definite operator.

\subsection{The Derivative of a Standard Matrix Function} \label{subsection:derivative}

Next, we present some classical results on the derivative of a standard matrix function.
See~\cite[Sec.~V.3]{Bhatia1997} for further details.

\begin{defn}[Divided Difference] \label{defn:divided-difference}
Let $f : I \to \R$ be a continuously differentiable function on an interval $I$
of the real line.  The \term{first divided difference} is the map
$f^{[1]} : \R^2 \to \R$ defined by
$$
f^{[1]}(\lambda, \mu) := \begin{dcases}
	f'(\lambda), & \lambda = \mu. \\
	\frac{f(\lambda) - f(\mu)}{\lambda - \mu}, & \lambda \neq \mu,
\end{dcases}
$$
We also require the Hermite representation of the divided difference:
\begin{equation} \label{eqn:hermite}
f^{[1]}(\lambda, \mu)
	= \int_0^1 f'(\tau \lambda + \bar{\tau} \mu) \idiff{\tau},
\end{equation}
where we have written $\bar{\tau} := 1-\tau$.
\end{defn}

The following result %
gives an explicit expression for the derivative of a standard
matrix function in terms of a divided difference.

\begin{prop}[Dalecki{\u\i}--Kre{\u\i}n Formula] \label{prop:derivative}
Let $f : I \to \R$ be a continuously differentiable function of an interval $I$
of the real line.  Suppose that $\mtx{A} \in \Sym^d(I)$ is a diagonal matrix
with $\mtx{A} = \diag(a_1, \dots, a_d)$.
The derivative ${\op D}f(\mtx{A}) \in \B(\M^d)$, and 
$$
{\op D}f(\mtx{A})(\mtx{H})
	= f^{[1]}(\mtx{A}) \odot \mtx{H}
	\quad\text{for $\mtx{H} \in \M^d$,}
$$
where $\odot$ denotes the Schur (i.e., componentwise) product and $f^{[1]}(\mtx{A})$ refers to the matrix of divided differences:
$$
\big[ f^{[1]}(\mtx{A}) \big]_{ij} = f^{[1]}(a_i, a_j)
\quad\text{for $i, j = 1, \dots, d$.}
$$
\end{prop}

\subsection{Operator Means} \label{subsubsection:concave-mean}

Our approach also relies on the concept of an operator mean.  The following definition is due to Kubo \& Ando~\cite{Kubo1980}.

\begin{defn}[Operator Mean]
Let $f : \R_{++} \to \R_{++}$ be an operator concave function that satisfies $f(1) = 1$.  Fix a natural number $d$.  Let $\op{S}$ and  $\op{T}$ be positive-definite operators in $\B(\M^d)$.  We define the mean of the operators:
$$
\op{M}_f(\op{S}, \op{T}) := \op{S}^{1/2} \cdot f(\op{S}^{-1/2} \op{T} \op{S}^{-1/2}) \cdot \op{S}^{1/2}
	\in \B(\M^d).
$$
When $\op{S}$ and $\op{T}$ commute, the formula simplifies to
$$
\op{M}_f(\op{S}, \op{T}) = \op{S} \cdot f(\op{T} \op{S}^{-1}).
$$
\end{defn}

\noindent
A few examples may be helpful.  The function $f(s) = (1+s)/2$ represents the usual arithmetic mean, the function $f(s) = s^{1/2}$ gives the geometric mean, and the function $f(s) = 2s/(1+s)$ yields the harmonic mean.  Operator means have a concavity property, which was established in the paper~\cite{Kubo1980}.

\begin{prop}[Operator Means are Concave] \label{prop:mean}
Let $f : \R_{++} \to \R_{++}$ be an operator monotone function with $f(1) = 1$.  Fix a natural number $d$.  Suppose that $\op{S}_1, \op{S}_2, \op{T}_1, \op{T}_2$ are positive-definite operators in $\B(\M^d)$.  Then
$$
\alpha \cdot \op{M}_f(\op{S}_1, \op{T}_1) +
\bar{\alpha} \cdot \op{M}_f(\op{S}_2, \op{T}_2)
\psdle \op{M}_f(\alpha \op{S}_1 + \bar{\alpha} \op{S}_2,
\alpha \op{T}_1 + \bar{\alpha} \op{T}_2)
$$
for $\alpha \in [0,1]$ and $\bar{\alpha} = 1 - \alpha$.
\end{prop}

\subsection{Entropy} \label{subsection:regular-entropy}

In this section, we demonstrate that the standard entropy function is a member of the $\Phi_{\infty}$ function class.

\begin{thm} \label{thm:entropy}
The function $\phi : t \mapsto t \log t - t$ is a member of the $\Phi_{\infty}$ class.
\end{thm}

\noindent
This result immediately implies Theorem~\ref{thm:membership}(1), which states that $t \mapsto t\log t$ belongs to $\Phi_{\infty}$.  Indeed, the matrix entropy class contains all affine functions and all finite sums of its elements.

Theorem~\ref{thm:entropy} follows easily from (deep) classical results because the variational representation of the standard entropy from Lemma~\ref{lem:lower-variation} is equivalent with the joint convexity of quantum relative entropy~\cite{Lindblad1973}.  Instead of pursuing this idea, we present an argument that parallels the approach we use to study the power function.  Some of the calculations below also appear in~\cite[Proof of Cor.~2.1]{Lieb1973}, albeit in compressed form.

\begin{proof}
Fix a positive integer $d$.  We plan to show that the function $\phi : t \mapsto t\log t - t$ is a
member of the class $\Phi_d$.  Evidently, $\phi$ is continuous and convex on $\R_+$, and it has two continuous derivatives on $\R_{++}$.  It remains to verify the concavity condition for the second derivative.  

Write $\psi(t) = \phi'(t) = \log t$, and let $\mtx{A} \in \Sym^d_{++}$.  Without loss of generality, we may choose a basis where $\mtx{A} = \diag(a_1, \dots, a_d)$.  The Dalecki{\u\i}--Kre{\u\i}n formula, Proposition~\ref{prop:derivative}, tells us
$$
\op{D} \psi(\mtx{A})(\mtx{H})
	= \psi^{[1]}(\mtx{A}) \odot \mtx{H}
	= \big[ \psi^{[1]}(a_i, a_j) \cdot h_{ij} \big]_{ij}.
$$
As an operator, the derivative acts by Schur multiplication.  This formula also makes it clear that the inverse of this operator acts by Schur multiplication:
$$
[\op{D} \psi(\mtx{A})]^{-1}(\mtx{H})
	= \bigg[ \frac{1}{\psi^{[1]}(a_i, a_j)} \cdot h_{ij} \bigg]_{ij}.
$$
Using the Hermite representation~\eqref{eqn:hermite} of the first divided difference of $t \mapsto \econst^t$, we find
$$
\frac{1}{\psi^{[1]}(\mu, \lambda)} = \frac{\lambda - \mu}{\log \lambda - \log \mu}
	= \int_0^1 \econst^{\tau \log \lambda + \bar{\tau} \log \mu} \idiff{\tau}
	= \int_0^1 \lambda^\tau \mu^{\bar{\tau}} \idiff{\tau}.
$$
The latter calculation assumes that $\mu \neq \lambda$; it extends to the case $\mu = \lambda$ because both sides of the identity are continuous.  As a consequence,
$$
[\op{D} \psi(\mtx{A})]^{-1}(\mtx{H})
	= \int_0^1 \bigg[ a_{i}^\tau h_{ij} a_{j}^{\bar{\tau}} \bigg]_{ij} \idiff{\tau}
	= \int_0^1 \mtx{A}^\tau \mtx{H} \mtx{A}^{\bar{\tau}} \idiff{\tau}
	= \int_0^1 (\mtx{A}^\tau \otimes \mtx{A}^{\bar{\tau}})(\mtx{H}) \idiff{\tau}.
$$
We discover the expression
\begin{equation} \label{eqn:entropy-almost}
[{\op D} \psi(\mtx{A}) ]^{-1} = \int_0^1 \mtx{A}^{\tau} \otimes \mtx{A}^{\bar{\tau}} \idiff{\tau}.
\end{equation}
This formula is correct for every positive-definite matrix.

For each $\tau \in [0,1]$, consider the operator monotone function $f : t \mapsto t^{\tau}$ defined on $\R_+$.  Since $f(1) = 1$, we can construct the operator mean $\op{M}_f$ associated with the function $f$.  Note that $\mtx{A} \otimes \Id$ and $\Id \otimes \mtx{A}$ are commuting positive operators.  Thus,
$$
\op{M}_f( \mtx{A} \otimes \Id, \Id \otimes \mtx{A} )
	= (\mtx{A} \otimes \Id)^{-1} f((\Id \otimes \mtx{A})(\mtx{A} \otimes \Id)^{-1})
	= \mtx{A}^\tau \otimes \mtx{A}^{\bar{\tau}}.
$$
The map $\mtx{A} \mapsto (\mtx{A} \otimes \Id, \Id \otimes \mtx{A})$ is linear, so
Proposition~\ref{prop:mean} guarantees that
$\mtx{A} \mapsto \mtx{A}^\tau \otimes \mtx{A}^{\bar{\tau}}$
is concave for each $\tau \in [0,1]$.  This result is usually called the Lieb Concavity Theorem~\cite[Thm.~IX.6.1]{Bhatia1997}.  Combine this fact with the integral representation~\eqref{eqn:entropy-almost} to reach the conclusion that
$\mtx{A} \mapsto [\op{D} \psi(\mtx{A})]^{-1}$ is a concave map on the cone $\Sym_{++}^d$ of positive-definite matrices.
\end{proof}

\subsection{Power Functions} \label{subsection:power}

In this section, we prove that certain power functions belong to the $\Phi_{\infty}$ function class.

\begin{thm} \label{thm:power}
For each $p \in [0,1]$, the function $\phi : t \mapsto t^{p+1}/(p+1)$ is a member of the $\Phi_{\infty}$ class.
\end{thm}

\noindent
This result immediately implies Theorem~\ref{thm:membership}(2), which states that $t \mapsto t^{p+1}$ belongs to the class $\Phi_{\infty}$.  Indeed, the matrix entropy class contains all positive multiples of its elements.

The proof of Theorem~\ref{thm:power} follows the same path as Theorem~\ref{thm:entropy},  but it is somewhat more involved.  First, we derive an expression for the function $\mtx{A} \mapsto [\op{D}\psi(\mtx{A})]^{-1}$ where $\psi = \phi'$.

\begin{lemma} \label{lem:rep}
Fix $p \in (0,1]$, and let $\psi(t) = t^p$ for $t \geq 0$.  For each matrix $\mtx{A} \in \Sym^d_+$,
\begin{equation} \label{eqn:formula}
[ \op{D}\psi(\mtx{A}) ]^{-1} = \frac{1}{p} \int_0^1
	(\tau \cdot \mtx{A}^p \otimes \Id + \bar{\tau} \cdot \Id \otimes \mtx{A}^p )^{(1-p)/p}
	\idiff{\tau},
\end{equation}
where $\bar{\tau} := 1 - \tau$.
\end{lemma}

\begin{proof}
As before, we may assume without loss of generality that the matrix $\mtx{A} = \diag(a_1, \dots, a_d)$.  Using the Dalecki{\u\i}--Kre{\u\i}n formula, Proposition~\ref{prop:derivative}, we see that
$$
[ \op{D}\psi(\mtx{A}) ]^{-1}(\mtx{H}) = \bigg[ \frac{1}{\psi^{[1]}(a_i, a_j)} \cdot h_{ij} \bigg].
$$
The Hermite representation~\eqref{eqn:hermite} of the first divided difference of $t \mapsto t^{1/p}$ gives
$$
\frac{1}{\psi^{[1]}(\mu, \lambda)}
	= \frac{\mu - \lambda}{\mu^p - \lambda^p}
	= \frac{1}{p} \int_0^1 (\tau \cdot \lambda^p + \bar{\tau} \cdot \mu^p )^{(1-p)/p} \idiff{\tau}
	=: g(\lambda, \mu).
$$
We use continuity to verify that the latter calculation remains valid when $\mu = \lambda$.
Using this function $g$, we can identify a compact representation of the operator:
$$
[ \op{D}\psi(\mtx{A}) ]^{-1}(\mtx{H})
	= \sum\nolimits_{ij} g(a_i, a_j) h_{ij} \mathbf{E}_{ij}
	= \bigg[\sum\nolimits_{ij} g(a_i, a_j) (\mathbf{E}_{ii} \otimes \mathbf{E}_{jj}) \bigg](\mtx{H}),
$$
where we write $\mathbf{E}_{ij}$ for the matrix with a one in the $(i, j)$ position and zeros elsewhere.  It remains to verify that the bracket coincides with the expression~\eqref{eqn:formula}.  Indeed,
\begin{align*}
\sum\nolimits_{ij} g(a_i, a_j) (\mathbf{E}_{ii} \otimes \mathbf{E}_{jj})
	&= \frac{1}{p} \int_0^1 \sum\nolimits_{ij}
	( \tau \cdot a_i^p + \bar{\tau} \cdot a_j^p )^{(1-p)/p} \, 
	(\mathbf{E}_{ii} \otimes \mathbf{E}_{jj}) \idiff{\tau} \\
	&= \frac{1}{p} \int_0^1 \bigg[ \sum\nolimits_{ij}
	(\tau \cdot a_i^p + \bar{\tau} \cdot a_j^p) (\mathbf{E}_{ii} \otimes \mathbf{E}_{jj})
	\bigg]^{(1-p)/p} \idiff{\tau} \\
	&= \frac{1}{p} \int_0^1 ( \tau \cdot \mtx{A}^p \otimes \Id +
	\bar{\tau} \cdot \Id \otimes \mtx{A}^p )^{(1-p)/p} \idiff{\tau}.
\end{align*}
The second relation follows from the definition of the standard operator function associated with $t \mapsto t^{(1-p)/p}$.  To confirm that the third line equals the second, expand the matrices $\mtx{A} = \sum_i a_i \mathbf{E}_{ii}$ and $\Id = \sum_j \mathbf{E}_{jj}$ and invoke the bilinearity of the tensor product.
\end{proof}

\begin{proof}[Proof of Theorem~\ref{thm:power}]
We are now prepared to prove that certain power functions belong to the $\Phi_{\infty}$ function class.  Fix an exponent $p \in [0,1]$, and let $d$ be a fixed positive integer.  We intend to show that the function $\phi(t) = t^{p+1}/(p+1)$ belongs to the $\Phi_d$ class.  When $p = 0$, the function $\phi$ is affine, so we may assume that $p > 0$.  It is clear that $\phi$ is continuous and convex on $\R_+$, and $\phi$ has two continuous derivatives on $\R_{++}$.  It remains to verify that the second derivative has the required concavity property.

Let $\psi(t) = \phi'(t) = t^p$ for $t \geq 0$, and consider a matrix $\mtx{A} \in \Sym^d_{++}$.  Lemma~\ref{lem:rep} demonstrates that
\begin{equation} \label{eqn:rep}
[ \op{D} \psi(\mtx{A}) ]^{-1}
	= \frac{1}{p} \int_0^1 ( \tau \cdot \mtx{A}^p \otimes \Id +
	\bar{\tau} \cdot \Id \otimes \mtx{A}^p )^{(1/p)(1-p)} \idiff{\tau},
\end{equation}
where we maintain the usage $\bar{\tau} := 1 - \tau$.  For each $\tau \in [0,1]$, the scalar function $a \mapsto \tau a + \bar{\tau}$ is operator monotone because it is affine and increasing.  On account of the result~\cite[Cor.~4.3]{Ando1979}, the function
$$
f : a \mapsto (\tau \cdot a^p + \bar{\tau})^{1/p}
$$
is also operator monotone.  A short calculation shows that $f(1) = 1$.  Therefore, we can use $f$ to construct an operator mean $\op{M}_f$.  Since $\mtx{A} \otimes \Id$ and $\Id\otimes \mtx{A}$ are commuting positive operators, we have
$$
\op{M}_f(\mtx{A} \otimes \Id, \Id \otimes \mtx{A}) = (\mtx{A} \otimes \Id)^{-1}
	f((\Id \otimes \mtx{A})(\mtx{A} \otimes \Id)^{-1} )
	= (\tau \cdot \mtx{A}^p \otimes \Id
	+ \bar{\tau} \cdot \Id \otimes \mtx{A}^p \big)^{1/p}.
$$
The map $\mtx{A} \mapsto (\mtx{A} \otimes \Id, \Id \otimes \mtx{A})$ is linear, so Proposition~\ref{prop:mean} ensures that
\begin{equation} \label{eqn:concave-almost}
\mtx{A} \mapsto (\tau \cdot \mtx{A}^p \otimes \Id
	+ \bar{\tau} \cdot \Id \otimes \mtx{A}^p )^{1/p}
\end{equation}
is a concave map.  

We are now prepared to check that~\eqref{eqn:rep} defines a concave operator.  Let $\mtx{S}, \mtx{T}$ be arbitrary positive-definite matrices, and choose $\alpha \in [0,1]$.  Suppose that $\mtx{Z}$ is the random matrix that takes value $\mtx{S}$ with probability $\alpha$ and value $\mtx{T}$ with probability $1 - \alpha$.  For each $\tau \in [0,1]$, we compute
\begin{align*}
\Expect \big[ (\tau \cdot \mtx{Z}^p \otimes \Id
	+ \bar{\tau} \cdot \Id \otimes \mtx{Z}^p )^{1/p} \big]^{1-p}
	&\psdle \big[ \Expect{} (\tau \cdot \mtx{Z}^p \otimes \Id
	+ \bar{\tau} \cdot \Id \otimes \mtx{Z}^p )^{1/p} \big]^{1-p} \\
	&\psdle \big[ \big(\tau \cdot (\Expect \mtx{Z})^p \otimes \Id
	+ \bar{\tau} \cdot \Id \otimes (\Expect \mtx{Z})^p \big)^{1/p} \big]^{1-p}.
\end{align*}
The first relation holds because $t \mapsto t^{1-p}$ is operator concave~\cite[Thm.~V.1.9 and Thm.~V.2.5]{Bhatia1997}.  To obtain the second relation, we apply the concavity property of the map~\eqref{eqn:concave-almost}, followed by the fact that $t \mapsto t^{1-p}$ is operator monotone~\cite[Thm.~V.1.9]{Bhatia1997}.  This calculation establishes the claim that
$$
\mtx{A} \mapsto \big(\tau \cdot \mtx{A}^p \otimes \Id
	+ \bar{\tau} \cdot \Id \otimes \mtx{A}^p \big)^{(1-p)/p}
$$
is concave on $\Sym_{++}^d$ for each $\tau \in [0,1]$.  In view of the integral representation~\eqref{eqn:rep}, we may conclude that $\mtx{A} \mapsto [\op{D} \psi(\mtx{A})]^{-1}$ is concave on the cone $\Sym^d_{++}$ of positive-definite matrices.
\end{proof}

\section{A Bounded Difference Inequality for Random Matrices} \label{section:concentration}

In this section, we prove Theorem~\ref{thm:tail-bound}, a bounded difference inequality for a random matrix whose distribution is invariant under signed permutation.  We begin with some preliminaries that support the proof, and we establish the main result in Section~\ref{subsection:proof-concentration}.

\subsection{Preliminaries} \label{subsection:rotation}

First, we describe how to compute the expectation of a function of a random matrix whose distribution is invariant under signed permutation.  See Definition~\ref{def:signed-permutation} for a reminder of what this requirement means.

\begin{lemma} \label{lem:invariant}
Let $f : I \to \R$ be a function on an interval $I$ of the real line.  Assume that $\mtx{X} \in \Sym^d(I)$ is a random matrix whose distribution is invariant under signed permutation.  Then
$$
\Expect f(\mtx{X}) = \ntr[ \Expect f(\mtx{X}) ] \cdot \Id.
$$
\end{lemma}

\begin{proof}
Let $\mtx{\Pi} \in \Sym^d$ be an arbitrary signed permutation matrix.  Observe that
\begin{equation} \label{eqn:sign-perm}
\Expect f(\mtx{X}) = \Expect f( \mtx{\Pi}^\adj \mtx{X} \mtx{\Pi} )
	= \mtx{\Pi}^\adj [\Expect f(\mtx{X})] \mtx{\Pi}.
\end{equation}
The first relation holds because the distribution of $\mtx{X}$ is invariant under conjugation by $\mtx{\Pi}$.  The second relation follows from the definition of a standard matrix function and the fact that $\mtx{\Pi}$ is unitary.  We may average~\eqref{eqn:sign-perm} over $\mtx{\Pi}$ drawn from the uniform distribution on the set of signed permutation matrices.  A direct calculation shows that the resulting matrix is diagonal, and its diagonal entries are identically equal to $\ntr[ \Expect f(\mtx{X}) ]$.
\end{proof}

We also require a trace inequality that is related to the mean value theorem.  This result specializes~\cite[Lem.~3.4]{Mackey2012}.

\begin{prop}[Mean Value Trace Inequality] \label{prop:mvt1}
Let $f : I \to \R$ be a function on an interval $I$ of the real line whose derivative $f'$ is convex.  For all $\mtx{A}, \mtx{B} \in \Sym^d(I)$,
$$
\ntr[ (\mtx{A} - \mtx{B}) (f(\mtx{A}) - f(\mtx{B})) ]
	\leq \frac{1}{2} \ntr [ (\mtx{A} - \mtx{B})^2 \cdot (f'(\mtx{A}) + f'(\mtx{B})) ].
$$
\end{prop}

\subsection{Proof of Theorem \ref{thm:tail-bound}} \label{subsection:proof-concentration}

The argument proceeds in three steps.  First, we present some elements of the matrix Laplace transform method.  Second, we use the subaddivity of matrix $\phi$-entropy to deduce a differential inequality for the trace moment generating function of the random matrix.  Finally, we explain how to integrate the differential inequality to obtain the concentration result.

\subsubsection{The Matrix Laplace Transform Method} \label{subsubsection:matrix-laplace}

We begin with a matrix extension of the moment generating function (mgf), which has played a major role in recent work on matrix concentration.

\begin{defn}[Trace Mgf]
Let $\mtx{Y}$ be a random Hermitian matrix.  The \term{normalized trace moment generating function} of $\mtx{Y}$ is defined as
$$
m(\theta) := m_{\mtx{Y}}(\theta) := \Expect \ntr \econst^{\theta \mtx{Y}}
\quad\text{for $\theta \in \R$.}
$$
The expectation need not exist for all values of $\theta$.
\end{defn}

The following proposition explains how the trace mgf can be used to study the maximum eigenvalue of a random Hermitian matrix~\cite[Prop.~3.1]{Tropp2011}.

\begin{prop}[Matrix Laplace Transform Method] \label{prop:laplace}
Let $\mtx{Y} \in \Sym^d$ be a random matrix with normalized trace mgf $m(\theta) := \ntr \econst^{\theta\mtx{Y}}$.  For each $t \in \R$,
$$
\Prob{\lambda_{\max}(\mtx{Y}) \geq t} \leq d \cdot \inf_{\theta > 0}
	\econst^{-\theta t + \log m(\theta)} \label{eqn:tail1}.
$$
\end{prop}

\subsubsection{A Differential Inequality for the Trace Mgf} \label{subsubsection:diff-ineq}

Suppose that $\mtx{Y} \in \Sym^d$ is a random Hermitian matrix that depends on a random vector $\vct{x} := (X_1, \dots, X_n)$.  We require the distribution of $\mtx{Y}$ to be invariant under signed permutations, and we insist that $\norm{\mtx{Y}}$ is bounded.  Without loss of generality, assume that $\mtx{Y}$ has zero mean.  Throughout the argument, we let the notation of Section~\ref{subsection:subadditivity} and Theorem~\ref{thm:tail-bound} prevail.

Let us explain how to use the subadditivity of matrix $\phi$-entropy to derive a differential inequality for the trace mgf.  Consider the function $\phi(t) = t \log t$, which belongs to the $\Phi_{\infty}$ class because of Theorem~\ref{thm:membership}(1).  Introduce the random positive-definite matrix $\mtx{Z} := \econst^{\theta \mtx{Y}}$, where $\theta > 0$.  We write out an expression for the matrix $\phi$-entropy of $\mtx{Z}$:
\begin{align}
H_{\phi}(\mtx{Z}) &= \Expect \ntr[ \phi(\mtx{Z}) - \phi(\Expect \mtx{Z})] \notag \\
	&= \Expect \ntr\big [ (\theta \mtx{Y}) \econst^{\theta \mtx{Y}} -
	\econst^{\theta \mtx{Y}} \log \Expect \econst^{\theta \mtx{Y}} \big] \notag \\
	&= \theta \cdot \Expect \ntr \big[ \mtx{Y} \econst^{\theta \mtx{Y}} \big]
	- (\Expect \ntr \econst^{\theta \mtx{Y}}) \log( \Expect \ntr \econst^{\theta \mtx{Y}})
	\notag \\
	&= \theta m'(\theta) - m(\theta) \log m(\theta). \label{eqn:left}
\end{align}
In the third line, we have applied Lemma~\ref{lem:invariant} to the logarithm in the second term, relying on the fact that $\mtx{Y}$ is invariant under signed permutations.  To reach the last line, we recognize that $m'(\theta) = \Expect \ntr( \mtx{Y}\econst^{\theta \mtx{Y}} )$.  We have used the boundedness of $\norm{\mtx{Y}}$ to justify this derivative calculation.

Corollary~\ref{cor:subadditivity-intro} provides an upper bound for the matrix $\phi$-entropy.
Define the derivative $\psi(t) = \phi'(t) = 1 + \log t$.  Then
\begin{align*}
H_{\phi}(\mtx{Z}) &\leq \frac{1}{2} \sum\nolimits_{i=1}^n
	\Expect \ntr \big[ (\mtx{Z} - \mtx{Z}_i') (\psi(\mtx{Z}) - \psi( \mtx{Z}_i')\big] \\
	&= \frac{\theta}{2} \sum\nolimits_{i=1}^n
	\Expect \ntr \big[ (\econst^{\theta \mtx{Y}} - \econst^{\theta \mtx{Y}_i'})
	(\mtx{Y} - \mtx{Y}_i') \big].
\end{align*}
Consider the function $f : t \mapsto \econst^{\theta t}$.  Its derivative $f' : t \mapsto \theta \econst^{\theta t}$ is convex because $\theta > 0$, so Proposition~\ref{prop:mvt1} delivers the bound
\begin{align*}
H_{\phi}(\mtx{Z})
&\leq \frac{\theta^2}{4} \sum\nolimits_{i=1}^n
	\Expect \ntr \big[ (\econst^{\theta \mtx{Y}} + \econst^{\theta \mtx{Y}_i'})
	(\mtx{Y} - \mtx{Y}_i')^2 \big] \\
&= \frac{\theta^2}{2} \sum\nolimits_{i=1}^n
	\Expect \ntr \big[ \econst^{\theta \mtx{Y}}
	(\mtx{Y} - \mtx{Y}_i')^2 \big] \\
&= \frac{\theta^2}{2} \sum\nolimits_{i=1}^n
	\Expect \ntr \big[ \econst^{\theta \mtx{Y}} \cdot
	\Expect[ (\mtx{Y} - \mtx{Y}_i')^2 \condl \vct{x} ] \big].
\end{align*}
The second relation follows from the fact that $\mtx{Y}$ and $\mtx{Y}_i'$ are exchangeable, conditional on $\vct{x}_{-i}$.  The last line is just the tower property of conditional expectation, combined with the observation that $\mtx{Y}$ is a function of $\vct{x}$.  To continue, we simplify the expression and make some additional bounds.
\begin{align}
H_{\phi}(\mtx{Z})
&\leq \frac{\theta^2}{2} 
	\Expect \ntr \bigg[ \econst^{\theta \mtx{Y}} \cdot
	\sum\nolimits_{i=1}^n \Expect[ (\mtx{Y} - \mtx{Y}_i')^2 \condl \vct{x} ] \bigg] \notag \\
&\leq \frac{\theta^2}{2} (\Expect \ntr \econst^{\theta \mtx{Y}})
	\norm{ \sum\nolimits_{i=1}^n \Expect[ (\mtx{Y} - \mtx{Y}_i')^2 \condl \vct{x} ] } \notag \\
&\leq \frac{\theta^2 V_{\mtx{Y}}}{2} \cdot m(\theta). \label{eqn:right}
\end{align}
The second relation follows from a standard trace inequality and the observation that $\econst^{\theta \mtx{Y}}$ is positive definite.  Last, we identify the variance measure $V_{\vct{Y}}$ defined in~\eqref{eqn:scalar-variance} and the trace mgf $m(\theta)$.

Combine the expression~\eqref{eqn:left} with the inequality~\eqref{eqn:right} to arrive at the estimate
\begin{equation} \label{eqn:diff-ineq}
\theta m'(\theta) - m(\theta) \log m(\theta)
	\leq \frac{\theta^2 V_{\mtx{Y}}}{2} \cdot m(\theta)
	\quad\text{for $\theta > 0$.}
\end{equation}
We can use this differential inequality to obtain bounds on the trace mgf $m(\theta)$.

\subsubsection{Solving the Differential Inequality} \label{subsubsection:integrate}

Rearrange the differential inequality~\eqref{eqn:diff-ineq} to obtain
\begin{equation} \label{eqn:exact-int} %
\frac{\diff{}}{\diff{\theta}} \bigg[\frac{\log m(\theta)}{\theta} \bigg]
= \frac{m'(\theta)}{\theta m(\theta)} - \frac{\log m(\theta)}{\theta^2}
	\leq \frac{V_{\mtx{Y}}}{2}.
\end{equation}
The l'H{\^ o}pital rule allows us to calculate the value of $\theta^{-1} \log m(\theta)$ at zero.  Since $m(0) = 1$,
$$
\lim_{\theta \to 0} \frac{\log m(\theta)}{\theta}
	= \lim_{\theta \to 0} \frac{m'(\theta)}{m(\theta)}
	= \lim_{\theta \to 0} \frac{ \Expect \ntr( \mtx{Y} \econst^{\theta \mtx{Y}} ) }
	{\Expect \ntr \econst^{\theta \mtx{Y}}}
	= \Expect \ntr \mtx{Y} = 0.
$$
This is where we use the hypothesis that $\mtx{Y}$ has mean zero.  Now,
we integrate~\eqref{eqn:exact-int} from zero to some positive value $\theta$ to find that the trace mgf satisfies
\begin{equation} \label{eqn:mgf-bd}
\frac{\log m(\theta)}{\theta} \leq \frac{\theta V_{\mtx{Y}}}{2}
\quad\text{when $\theta > 0$.}
\end{equation}
The approach in this section is usually referred to as the Herbst argument~\cite{Ledoux1997}.

\subsubsection{The Laplace Transform Argument}

We are now prepared to finish the argument.  Combine the matrix Laplace transform method, Proposition~\ref{prop:laplace}, with the trace mgf bound~\eqref{eqn:mgf-bd} to reach
\begin{equation} \label{eqn:upper-tail}
\Prob{ \lambda_{\max}(\mtx{Y}) \geq t }	
	\leq d \cdot \inf_{\theta > 0} \econst^{-\theta t + \log m(\theta)}
	\leq d \cdot \inf_{\theta > 0}  \econst^{-\theta t + \theta^2 V_{\mtx{Y}} / 2}
	= d \cdot \econst^{-t^2/(2 V_{\mtx{Y}})}.
\end{equation}
To obtain the result for the minimum eigenvalue, we note that
$$
\Prob{ \lambda_{\min}(\mtx{Y}) \leq -t }
	= \Prob{ \lambda_{\max}(-\mtx{Y}) \geq t }
	\leq d \cdot \econst^{-t^2/(2 V_{\mtx{Y}})}.
$$
The inequality follows when we apply~\eqref{eqn:upper-tail} to the random matrix $-\mtx{Y}$.  This completes the proof of Theorem~\ref{thm:tail-bound}.

\section{Moment Inequalities for Random Matrices with Bounded Differences} \label{section:moment}

In this section, we prove Theorem~\ref{thm:moment-bound-intro}, which gives information about the moments of a random matrix that satisfies a kind of self-bounding property.

\begin{proof}[Proof of Theorem~\ref{thm:moment-bound-intro}]
Fix a number $q \in \{2, 3, 4, \dots\}$.  Suppose that $\mtx{Y} \in \Sym^d_+$ is a random positive-semidefinite matrix that depends on a random vector $\vct{x} := (X_1, \dots, X_n)$.  We require the distribution of $\mtx{Y}$ to be invariant under signed permutations, and we assume that $\Expect( \norm{\mtx{Y}}^q ) < \infty$.  The notation of Section~\ref{subsection:subadditivity} and Theorem~\ref{thm:moment-bound-intro} remains in force.

Let us explain how the subadditivity of matrix $\phi$-entropy leads to a bound on the $q$th trace moment of $\mtx{Y}$.  Consider the power function $\phi(t) = t^{q/(q-1)}$.  Theorem~\ref{thm:power} ensures that $\phi \in \Phi_{\infty}$ because $q/(q-1) \in (1, 2]$.  Introduce the random positive-semidefinite matrix $\mtx{Z} := \mtx{Y}^{q-1}$.  Then
\begin{align}
H_{\phi}(\mtx{Z}) &= \Expect \ntr\big[ \phi(\mtx{Z}) - \phi(\Expect \mtx{Z}) \big] \notag \\
	&= \Expect \ntr( \mtx{Y}^q ) - \ntr\big[(\Expect(\mtx{Y}^{q-1}))^{q/(q-1)} \big] \notag \\
	&= \Expect \ntr( \mtx{Y}^q )  - \big[\Expect \ntr( \mtx{Y}^{q-1}) \big]^{q/(q-1)}.
	\label{eqn:moment-left}
\end{align}
The transition to the last line requires Lemma~\ref{lem:invariant}.

Corollary~\ref{cor:subadditivity-intro} provides an upper bound for the matrix $\phi$-entropy.  Define the derivative $\psi(t) = \phi'(t) = (q/(q-1)) \cdot t^{1/(q-1)}$.  We have
\begin{align*}
H_{\phi}(\mtx{Z}) &\leq \frac{1}{2} \sum\nolimits_{i=1}^n
	\Expect \ntr\big[ (\mtx{Z} - \mtx{Z}_i')( \psi(\mtx{Z}) - \psi(\mtx{Z}_i') ) \big] \\
	&= \frac{q}{2(q-1)} \sum\nolimits_{i=1}^n
	\Expect \ntr\big[ \big(\mtx{Y}^{q-1} - (\mtx{Y}_i')^{q-1} \big)
	( \mtx{Y} - \mtx{Y}_i' ) \big]
\end{align*}
The function $f : t \mapsto t^{q-1}$ has the derivative $f' : t \mapsto (q-1) t^{q-2}$, which is convex because $q \in\{ 2, 3, 4, \dots \}$.  Therefore, the mean value trace inequality, Proposition~\ref{prop:mvt1}, delivers the bound
\begin{align*}
H_{\phi}(\mtx{Z}) &\leq \frac{q}{4} \sum\nolimits_{i=1}^n
	\Expect \ntr\big[ \big(\mtx{Y}^{q-2} + (\mtx{Y}_i')^{q-2} \big)
	( \mtx{Y} - \mtx{Y}_i' )^2 \big] \\
	&= \frac{q}{2} \sum\nolimits_{i=1}^n
	\Expect \ntr\big[ \mtx{Y}^{q-2} ( \mtx{Y} - \mtx{Y}_i' )^2 \big] \\
	&= \frac{q}{2} \sum\nolimits_{i=1}^n
	\Expect \ntr\big[ \mtx{Y}^{q-2} \Expect[(\mtx{Y}-\mtx{Y}_i')^2 \condl \vct{x} ] \big].
\end{align*}
The second identity holds because $\mtx{Y}$ and $\mtx{Y}_{i}'$ are exchangeable, conditional on $\vct{x}_{-i}$.  The last line follows from the tower property of conditional expectation.  We simplify this expression as follows.
\begin{align}
H_{\phi}(\mtx{Z}) &\leq
	\frac{q}{2} \Expect \ntr\bigg[ \mtx{Y}^{q-2} \cdot
	\sum\nolimits_{i=1}^n \Expect[(\mtx{Y}-\mtx{Y}_i')^2 \condl \vct{x} ] \bigg] \notag \\
	&\leq \frac{q}{2} \Expect \ntr\big[ \mtx{Y}^{q-2} \cdot c\mtx{Y} \big] \notag \\
	&= \frac{cq}{2} \Expect \ntr(\mtx{Y}^{q-1}). \label{eqn:moment-right}
\end{align}
The second inequality derives from the hypothesis~\eqref{eqn:matrix-variance} that $\mtx{V}_{\mtx{Y}} \psdle c\mtx{Y}$.  Note that this bound requires the fact that $\mtx{Y}^{q-2}$ is positive semidefinite.

Combine the expression~\eqref{eqn:moment-left} for the matrix $\phi$-entropy with the upper bound~\eqref{eqn:moment-right} to achieve the estimate
$$
\Expect \ntr( \mtx{Y}^q ) - \big[\Expect \ntr( \mtx{Y}^{q-1}) \big]^{q/(q-1)}
	\leq \frac{cq}{2} \Expect \ntr(\mtx{Y}^{q-1}).
$$
Rewrite this bound, and invoke the numerical fact $1 + aq \leq (1+a)^q$ to obtain
\begin{align*}
\Expect \ntr( \mtx{Y}^q )
	&\leq \big[\Expect \ntr( \mtx{Y}^{q-1}) \big]^{q/(q-1)} \left(
	1 + \frac{cq/2}{\big[ \Expect \ntr(\mtx{Y}^{q-1}) \big]^{1/q-1}} \right) \\
	&\leq \big[\Expect \ntr( \mtx{Y}^{q-1}) \big]^{q/(q-1)} \left(
	1 + \frac{c/2}{\big[ \Expect \ntr(\mtx{Y}^{q-1}) \big]^{1/q-1}} \right)^q.
\end{align*}
Extract the $q$th root from both sides to reach
$$
\big[ \Expect \ntr( \mtx{Y}^q ) \big]^{1/q}
	\leq \big[\Expect \ntr( \mtx{Y}^{q-1}) \big]^{1/(q-1)} + \frac{c}{2}.
$$
We have compared the $q$th trace moment of $\mtx{Y}$ with the $(q-1)$th trace moment.  Proceeding by iteration, we arrive at
$$
\big[ \Expect \ntr( \mtx{Y}^q ) \big]^{1/q}
	\leq \Expect \ntr \mtx{Y} + \frac{q-1}{2} \cdot c.
$$
This observation completes the proof of Theorem~\ref{thm:moment-bound-intro}.
\end{proof}

\appendix

\section{Lemma~\ref{lem:lower-variation}, The General Case} \label{app:general}

In this appendix, we explain how to prove Lemma~\ref{lem:lower-variation} in full generality.  The argument calls for a simple but powerful result, known as the generalized Klein inequality~\cite[Prop.~3]{Petz1994}, which allows us to lift a large class of scalar inequalities to matrices.

\begin{prop}[Generalized Klein Inequality] \label{prop:gen-klein}
For each $k = 1, \dots, n$, suppose that $f_k : I_1 \to \R$ and $g_k : I_2 \to \R$ are functions on intervals $I_1$ and $I_2$ of the real line.  Suppose that
$$
\sum\nolimits_{k=1}^n f_k(a) \, g_k(b) \geq 0 \quad\text{for all $a \in I_1$ and $b \in I_2$.}
$$
Then, for each natural number $d$, 
$$
\sum\nolimits_{k=1}^n \ntr[ f_k(\mtx{A}) \, g_k(\mtx{B}) ] \geq 0
\quad\text{for all $\mtx{A} \in \Sym^d(I_1)$ and $\mtx{B} \in \Sym^d(I_2)$.}
$$
\end{prop}

\begin{proof}[Proof of Lemma~\ref{lem:lower-variation}, General Case]
We retain the notation from Lemma~\ref{lem:lower-variation}.  In particular, we assume that $\mtx{Z}$ is a random positive-definite matrix for which $\norm{ \mtx{Z} }$ and $\norm{\phi(\mtx{Z})}$ are both integrable.  We also assume that $\mtx{T}$ is a random positive-definite matrix with $\norm{\mtx{T}}$ and $\norm{\phi(\mtx{T})}$ integrable.

For $n \in \mathbb{N}$, define the function $l_n(a) := (a \vee 1/n)\wedge n$, where $\vee$ denotes the maximum operator and $\wedge$ denotes the minimum operator.  Consider the random matrices $\mtx{Z}_n := l_n(\mtx{T})$ and $\mtx{T}_k := l_k(\mtx{T})$ for each $k, n \in \mathbb{N}$.  These matrices have eigenvalues that are bounded and bounded away from zero, so these entities satisfy the inequality~\eqref{eqn:ineq} we have already established.
$$
H_{\phi}(\mtx{Z}_n) \geq \Expect \ntr\big[
	(\psi(\mtx{T}_k) - \psi(\Expect \mtx{T}_k))(\mtx{Z}_n - \mtx{T}_k)
	+ \Expect \phi(\mtx{T}_k - \phi(\Expect \mtx{T}_k) \big].
$$
Rearrange the terms in this inequality to obtain
\begin{equation} \label{eqn:ineq-kn}
\Expect \ntr \mtx{\Gamma}(\mtx{Z}_n, \mtx{T}_k) \geq
	\ntr\big[ - \psi(\Expect \mtx{T}_k) (\Expect \mtx{Z}_n - \Expect \mtx{T}_k)
	- \phi(\Expect \mtx{T}_k) + \phi(\Expect \mtx{Z}_n) \big],
\end{equation}
where we have introduced the function
$$
\mtx{\Gamma}(\mtx{A}, \mtx{B}) := \phi(\mtx{A}) - \phi(\mtx{B}) - (\mtx{A} - \mtx{B}) \psi(\mtx{B})
\quad\text{for $\mtx{A}, \mtx{B} \in \Sym^d_{++}$.}
$$
To complete the proof of Lemma~\ref{lem:lower-variation}, we must develop the bound
\begin{equation} \label{eqn:lim-bound}
\Expect \ntr \mtx{\Gamma}(\mtx{Z}, \mtx{T})
	\geq \ntr\big[ - \psi(\Expect \mtx{T}) (\Expect \mtx{Z} - \Expect \mtx{T})
	- \phi(\Expect \mtx{T}) + \phi(\Expect \mtx{Z}) \big]
\end{equation}
by driving $k, n \to \infty$ in~\eqref{eqn:ineq-kn}.

Let us begin with the right-hand side of~\eqref{eqn:ineq-kn}.  We have the sure limit $\mtx{Z}_n \to \mtx{Z}$.  Therefore, the Dominated Convergence Theorem guarantees that $\Expect \mtx{Z}_n \to \Expect \mtx{Z}$ because $\norm{\mtx{Z}}$ is integrable and $\norm{\mtx{Z}_n} \leq \norm{\mtx{Z}}$.  Likewise, $\Expect \mtx{T}_k \to \Expect \mtx{T}$.  The functions $\phi$ and $\psi$ are continuous, so the limit of the right-hand side of~\eqref{eqn:ineq-kn} satisfies
\begin{multline} \label{eqn:ineq-kn-rhs}
\ntr\big[ - \psi(\Expect \mtx{T}_k) (\Expect \mtx{Z}_n - \Expect \mtx{T}_k)
	- \phi(\Expect \mtx{T}_k) + \phi(\Expect \mtx{Z}_n) \big] \\
	\to \ntr\big[ - \psi(\Expect \mtx{T}) (\Expect \mtx{Z} - \Expect \mtx{T})
	- \phi(\Expect \mtx{T}) + \phi(\Expect \mtx{Z}) \big].
\end{multline}
This expression coincides with the right-hand side of~\eqref{eqn:lim-bound}.

Taking the limit of the left-hand side of~\eqref{eqn:ineq-kn} is more involved because the function $\psi$ may grow quickly at zero and infinity.  We accomplish our goal in two steps.  First, we take the limit as $n \to \infty$.  Afterward, we take the limit as $k \to \infty$.

Introduce the nonnegative function
$$
\gamma(z,t) := \phi(z) - \phi(t) - (z-t) \psi(t)
\quad\text{for $z, t > 0$.}
$$
Boucheron et al.~\cite[p.~525]{Boucheron2005} establish that
\begin{equation} \label{eqn:gamma-sum}
\gamma(l_n(z), l_k(t)) \leq \gamma(1, l_k(t)) + \gamma(z, l_k(t))
\quad\text{for $z, t > 0$.}
\end{equation}
The generalized Klein inequality, Proposition~\ref{prop:gen-klein}, can be applied (with due diligence) to extend~\eqref{eqn:gamma-sum} to matrices.  In particular,
$$
\ntr \mtx{\Gamma}( \mtx{Z}_n , \mtx{T}_k )
	= \ntr \mtx{\Gamma}( l_n(\mtx{Z}), l_k(\mtx{T} ))
	\leq \ntr[ \mtx{\Gamma}(\Id, l_k(\mtx{T})) + \mtx{\Gamma}(\mtx{Z}, l_k(\mtx{T}))]
	= \ntr[ \mtx{\Gamma}(\Id, \mtx{T}_k) + \mtx{\Gamma}(\mtx{Z}, \mtx{T}_k)].
$$
Observe that the right-hand side of this inequality is integrable.  Indeed, all of the quantities involving $\mtx{T}_k$ are uniformly bounded because the eigenvalues of $\mtx{T}_k$ fall in the range $[k^{-1}, k]$ and the functions $\phi$ and $\psi$ are continuous on this interval.  The terms involving $\mtx{Z}$ may not be bounded, but they are integrable because $\norm{\mtx{Z}}$ and $\norm{\phi(\mtx{Z})}$ are integrable.  We may now apply the Dominated Convergence Theorem to take the limit:
\begin{equation} \label{eqn:n-lim}
\Expect \ntr \mtx{\Gamma}(\mtx{Z}_n, \mtx{T}_k)
	\to \Expect \ntr \mtx{\Gamma}(\mtx{Z}, \mtx{T}_k)
	\quad\text{as $n \to \infty$,}
\end{equation}
where we rely again on the sure limit $\mtx{Z}_n \to \mtx{Z}$ as $n \to \infty$.

Boucheron et al.~also establish that
$$
\gamma(z, l_k(t)) \leq \gamma(z, 1) + \gamma(z, t)
\quad\text{for $z, t > 0$.}
$$
The generalized Klein inequality, Proposition~\ref{prop:gen-klein}, ensures that
$$
\ntr \mtx{\Gamma}( \mtx{Z} , \mtx{T}_k ) \leq
	\ntr[ \mtx{\Gamma}(\mtx{Z}, \Id) + \mtx{\Gamma}(\mtx{Z}, \mtx{T}) ].
$$
We may assume that the second term on the right-hand side is integrable or else the desired inequality~\eqref{eqn:lim-bound} would be vacuous.  The first term is integrable because $\norm{ \mtx{Z} }$ and $\norm{\phi(\mtx{Z})}$ are integrable.  Therefore, we may apply the Dominated Convergence Theorem:
\begin{equation} \label{eqn:k-lim}
\Expect \ntr \mtx{\Gamma}(\mtx{Z}, \mtx{T}_k)
	\to \Expect \ntr \mtx{\Gamma}(\mtx{Z}, \mtx{T})
	\quad\text{as $k \to \infty$,}
\end{equation}
where we rely again on the sure limit $\mtx{T}_k \to \mtx{T}$ as $k \to \infty$.

In summary, the limits~\eqref{eqn:n-lim} and~\eqref{eqn:k-lim} provide that $\Expect \ntr \mtx{\Gamma}(\mtx{Z}_n, \mtx{T}_k) \to \Expect \ntr \mtx{\Gamma}(\mtx{Z}, \mtx{T})$ as $k, n \to \infty$.  In view of the limit~\eqref{eqn:ineq-kn-rhs}, we have completed the proof of~\eqref{eqn:lim-bound}.
\end{proof}

\section*{Acknowledgments}

RYC and JAT are with the Department of Computing and Mathematical Sciences, California Institute of Technology.  JAT gratefully acknowledges support from ONR awards N00014-08-1-0883 and N00014-11-1002, AFOSR award FA9550-09-1-0643, and a Sloan Research Fellowship.  JAT also wishes to thank the Moore Foundation.

\bibliographystyle{alpha}
\bibliography{ref}

\newcommand{\etalchar}[1]{$^{#1}$}
\begin{thebibliography}{BBLM05}

\bibitem[And79]{Ando1979}
T.~Ando.
\newblock Concavity of certain maps on positive definite matrices and
  applications to {H}adamard products.
\newblock {\em Linear Algebra Appl.}, 26:203--241, 1979.

\bibitem[ARR12]{Ahmed2012}
A.~Ahmed, B.~Recht, and J.~Romberg.
\newblock Blind deconvolution using convex programming.
\newblock Available at \url{arXiv.org/abs/1211.5608}, Nov. 2012.

\bibitem[AW02]{Ahlswede2002}
R.~Ahlswede and A.~Winter.
\newblock Strong converse for identification via quantum channels.
\newblock {\em IEEE Trans. Inform. Theory}, 48(3):569--579, 2002.

\bibitem[BBLM05]{Boucheron2005}
S.~Boucheron, O.~Bousquet, G.~Lugosi, and P.~Massart.
\newblock Moment inequalities for functions of independent random variables.
\newblock {\em Ann. Probab.}, 33(2):514--560, 2005.

\bibitem[BG12]{Boutsidis2012}
C.~Boutsidis and A.~Gittens.
\newblock Improved matrix algorithms via the subsampled randomized {H}adamard
  transform.
\newblock Available at \url{arXiv.org/abs/1204.0062}, 2012.

\bibitem[Bha97]{Bhatia1997}
R.~Bhatia.
\newblock {\em Matrix Analysis}, volume 169 of {\em Graduate Texts in
  Mathematics}.
\newblock Springer-Verlag, New York, 1997.

\bibitem[Bha07]{Bhatia2007}
R.~Bhatia.
\newblock {\em Positive Definite Matrices}.
\newblock Princeton Series in Applied Mathematics. Princeton University Press,
  Princeton, NJ, 2007.

\bibitem[BL98]{Bobkov1998}
S.~G. Bobkov and M.~Ledoux.
\newblock On modified logarithmic {S}obolev inequalities for {B}ernoulli and
  {P}oisson measures.
\newblock {\em J. Funct. Anal.}, 156(2):347--365, 1998.

\bibitem[BLM03]{Boucheron2003}
S.~Boucheron, G.~Lugosi, and P.~Massart.
\newblock Concentration inequalities using the entropy method.
\newblock {\em Ann. Probab.}, 31(3):1583--1614, 2003.

\bibitem[BLM13]{BLM13:Concentration-Inequalities}
S.~Boucheron, G.~Lugosi, and P.~Massart.
\newblock {\em Concentration inequalities}.
\newblock Oxford University Press, 2013.

\bibitem[Bou02]{Bousquet2002}
O.~Bousquet.
\newblock A {B}ennett concentration inequality and its application to suprema
  of empirical processes.
\newblock {\em C. R. Math. Acad. Sci. Paris}, 334(6):495--500, 2002.

\bibitem[Car10]{Carlen2010}
E.~Carlen.
\newblock Trace inequalities and quantum entropy: an introductory course.
\newblock In {\em Entropy and the quantum}, volume 529 of {\em Contemp. Math.},
  pages 73--140. Amer. Math. Soc., Providence, RI, 2010.

\bibitem[CCT12]{Chaudhuri2012}
K.~Chaudhuri, F.~Chung, and A.~Tsiatas.
\newblock Spectral clustering of graphs with general degrees in the extended
  planted partition model.
\newblock {\em Journal of Machine Learning Research 2012}, pages 1--23, 2012.

\bibitem[CDL13]{Cohen2011}
A.~Cohen, M.~A. Davenport, and D.~Leviatan.
\newblock On the stability and accuracy of least squares approximations.
\newblock {\em Found. Comput. Math.}, 2013.

\bibitem[Cha04]{Chafai2004}
D.~Chafa{\"{\i}}.
\newblock Entropies, convexity, and functional inequalities: on
  {$\Phi$}-entropies and {$\Phi$}-{S}obolev inequalities.
\newblock {\em J. Math. Kyoto Univ.}, 44(2):325--363, 2004.

\bibitem[Cha07]{Chatterjee2007}
S.~Chatterjee.
\newblock Stein's method for concentration inequalities.
\newblock {\em Probab. Theory Related Fields}, 138(1-2):305--321, 2007.

\bibitem[Cha08]{Chatterjee2008}
S.~Chatterjee.
\newblock {\em Concentration inequalities with exchangeable pairs}.
\newblock PhD thesis, Stanford University, Palo Alto, Feb. 2008.

\bibitem[DZ11]{Drineas2011}
P.~Drineas and A.~Zouzias.
\newblock A note on element-wise matrix sparsification via a matrix-valued
  {B}ernstein inequality.
\newblock {\em Inform. Process. Lett.}, 111:385--389, 2011.

\bibitem[Han13]{Han13:Convexity-Residual}
F.~Hansen.
\newblock Convexity of the residual entropy.
\newblock Available at \url{arXiv.org/abs/1305.1720}, May 2013.

\bibitem[HOZ01]{HOZ01:Upper-Bound}
W.~Hebisch, R.~Olkiewicz, and B.~Zegarlinski.
\newblock On upper bound for the quantum entropy.
\newblock {\em Linear Algebra Appl.}, 329:89--96, 2001.

\bibitem[JX03]{JX03:Noncommutative-Burkholder}
M.~Junge and Q.~Xu.
\newblock Noncommutative {B}urkholder/{R}osenthal inequalities.
\newblock {\em Ann. Probab.}, 31(2):948--995, 2003.

\bibitem[JX08]{JX08:Noncommutative-Burkholder-II}
M.~Junge and Q.~Xu.
\newblock Noncommutative {B}urkholder/{R}osenthal inequalities {II}:
  {A}pplications.
\newblock {\em Israel J. Math.}, 167:227--282, 2008.

\bibitem[JZ11]{JZ11:Noncommutative-Bennett}
M.~Junge and Q.~Zheng.
\newblock Noncommutative {B}ennett and {R}osenthal inequalities.
\newblock Available at \url{arXiv.org/abs/1111.1027}, Nov. 2011.

\bibitem[KA80]{Kubo1980}
F.~Kubo and T.~Ando.
\newblock Means of positive linear operators.
\newblock {\em Math. Ann.}, 246(3):205--224, 1979/80.

\bibitem[Kol11]{Koltchinskii2011}
V.~Koltchinskii.
\newblock Von {N}eumann entropy penalization and low-rank matrix estimation.
\newblock {\em Ann. Statist.}, 39(6):2936--2973, 2011.

\bibitem[Led99]{Ledoux1997}
M.~Ledoux.
\newblock Concentration of measure and logarithmic {S}obolev inequalities.
\newblock In {\em S\'eminaire de {P}robabilit\'es, {XXXIII}}, volume 1709 of
  {\em Lecture Notes in Math.}, pages 120--216. Springer, Berlin, 1999.

\bibitem[Led01]{Ledoux2001}
M.~Ledoux.
\newblock {\em The Concentration of Measure Phenomenon}, volume~89 of {\em
  Mathematical Surveys and Monographs}.
\newblock American Mathematical Society, Providence, RI, 2001.

\bibitem[Led97]{Ledoux1996}
M.~Ledoux.
\newblock On {T}alagrand's deviation inequalities for product measures.
\newblock {\em ESAIM Probab. Statist.}, 1:63--87 (electronic), 1995/97.

\bibitem[Lie73]{Lieb1973}
E.~H. Lieb.
\newblock Convex trace functions and the {W}igner-{Y}anase-{D}yson conjecture.
\newblock {\em Advances in Math.}, 11:267--288, 1973.

\bibitem[Lin73]{Lindblad1973}
G.~Lindblad.
\newblock Entropy, information, and quantum measurements.
\newblock {\em Commun. Math. Phys.}, 33:305--322, 1973.

\bibitem[LO00]{Latala2000}
R.~Lata{\l}a and K.~Oleszkiewicz.
\newblock Between {S}obolev and {P}oincar\'e.
\newblock In {\em Geometric aspects of functional analysis}, volume 1745 of
  {\em Lecture Notes in Math.}, pages 147--168. Springer, Berlin, 2000.

\bibitem[LR73]{LR73:Proof-Strong}
E.~H. Lieb and M.~B. Ruskai.
\newblock Proof of the strong subadditivity of quantum-mechanical entropy.
\newblock {\em J. Math. Phys.}, 14(12):1938--1941, 1973.

\bibitem[Mas00a]{Massart2000}
P.~Massart.
\newblock About the constants in {T}alagrand's concentration inequalities for
  empirical processes.
\newblock {\em Ann. Probab.}, 28(2):863--884, 2000.

\bibitem[Mas00b]{Massart2000a}
P.~Massart.
\newblock Some applications of concentration inequalities to statistics.
\newblock {\em Ann. Fac. Sci. Toulouse Math. (6)}, 9(2):245--303, 2000.
\newblock Probability theory.

\bibitem[Min12]{Minsker2012}
S.~Minsker.
\newblock On some extensions of {B}ernstein's inequality for self-adjoint
  operators.
\newblock Available at \url{arXiv.org/abs/1112.5448}, Jan. 2012.

\bibitem[MJC{\etalchar{+}}12]{Mackey2012}
L.~Mackey, M.~I. Jordan, R.~Y. Chen, B.~Farrell, and J.~A. Tropp.
\newblock Matrix concentration inequalities via the method of exchangeable
  pairs.
\newblock Available at \url{arXiv.org/abs/1201.6002}, 2012.

\bibitem[MTJ11]{Mackey2011}
L.~Mackey, A.~Talwalkar, and M.~I. Jordan.
\newblock Divide-and-conquer matrix factorization.
\newblock Available at \url{arXiv.org/abs/1107.0789}, 2011.

\bibitem[Oli09]{Oliveira2009}
R.~I. Oliveira.
\newblock Concentration of the adjacency matrix and of the {L}aplacian in
  random graphs with independent edges.
\newblock Available at \url{arXiv.org/abs/0911.0600}, 2009.

\bibitem[Oli10]{Oliveira2010}
R.~I. Oliveira.
\newblock Sums of random {H}ermitian matrices and an inequality by {R}udelson.
\newblock {\em Electron. Commun. Probab.}, 15:203--212, 2010.

\bibitem[Pet94]{Petz1994}
D~Petz.
\newblock A survey of certain trace inequalities.
\newblock In {\em Functional analysis and operator theory ({W}arsaw, 1992)},
  volume~30 of {\em Banach Center Publ.}, pages 287--298. Polish Acad. Sci.,
  Warsaw, 1994.

\bibitem[PMT13]{Paulin2013}
D.~Paulin, L.~Mackey, and J.~A. Tropp.
\newblock Deriving matrix concentration inequalities from kernel couplings.
\newblock Available at \url{arXiv.org/abs/1305.0612}, May 2013.

\bibitem[Rec11]{Recht2011}
B.~Recht.
\newblock A simpler approach to matrix completion.
\newblock {\em J. Mach. Learn. Res.}, 12:3413--3430, 2011.

\bibitem[Rio01]{Rio2001}
E.~Rio.
\newblock In\'egalit\'es de concentration pour les processus empiriques de
  classes de parties.
\newblock {\em Probab. Theory Related Fields}, 119(2):163--175, 2001.

\bibitem[Tal91]{Talagrand1991}
M.~Talagrand.
\newblock A new isoperimetric inequality for product measure, and the
  concentration of measure phenomenon.
\newblock In {\em Israel Seminar (GAFA)}, volume 1469 of {\em Lecture Notes in
  Math}. Springer-Verlag, 1991.

\bibitem[TBSR12]{Tang2012}
G.~Tang, B.~N. Bhaskar, P.~Shah, and B.~Recht.
\newblock Compressed sensing off the grid.
\newblock Available at \url{arXiv.org/abs/1207.6053}, 2012.

\bibitem[Tro11]{Tropp2011}
J.~A. Tropp.
\newblock Freedman's inequality for matrix martingales.
\newblock {\em Electron. Commun. Probab.}, 16:262--270, 2011.

\bibitem[Tro12a]{Tro11:Improved-Analysis}
J.~A. Tropp.
\newblock Improved analysis of the subsampled randomized {H}adamard transform.
\newblock {\em Adv. Adapt. Data Anal.}, 3(1--2):115--126, 2012.

\bibitem[Tro12b]{Tropp2012}
J.~A. Tropp.
\newblock User-friendly tail bounds for sums of random matrices.
\newblock {\em Found. Comput. Math.}, 12(4):389--434, 2012.

\bibitem[Tro12c]{Tro12:User-Friendly}
J.~A. Tropp.
\newblock User-friendly tools for random matrices: An introduction.
\newblock Available at
  \url{users.cms.caltech.edu/~jtropp/notes/Tro12-User-Friendly-Tools-NIPS.pdf},
  Dec. 2012.

\end{thebibliography}
\end{document}